\documentclass[a4paper,UKenglish]{lipics-v2016}
\usepackage{microtype}%if unwanted, comment out or use option "draft"
\usepackage{multirow}
\usepackage[utf8]{inputenc}
\usepackage[inference]{semantic}
\usepackage{verbatim}
\usepackage{stmaryrd}
\usepackage{enumerate}
\usepackage{extarrows}
\usepackage{hyperref}
\DeclareMathAlphabet{\mathbbold}{U}{bbold}{m}{n}
\newcommand{\reg}[1]{\textsc{#1-reg}}
\newcommand{\Rd}[1]{\textsc{Rd}(#1)}
\newcommand{\enc}[1]{\llbracket #1\rrbracket}

\newcommand{\act}[2][]{{\stackrel{#2}{\longrightarrow}}^{#1}}
\newcommand{\Act}[1]{{\stackrel{#1}{\longrightarrow}}{^{*}}}
\newcommand{\actx}[2][]{\xrightarrow{#2}^{#1}}

\newcommand{\norm}[1]{\lVert #1 \rVert}
\newcommand{\bnorm}[1]{\lVert #1 \rVert_b}

\newcommand{\beq}{\simeq}
\newcommand{\weq}{\approx}

\newlength{\probwidth}
\newcommand{\nprob}[4][9]{
\begin{center}\normalfont\fbox{
\addtolength{\probwidth}{#1cm}\parbox{\probwidth}{\textsc{#2}\\\hspace*{1.5em}
\begin{tabular}[t]{rp{#1cm}}%\textsc{#2} \
\textit{Instance:}&#3.\\\textit{Question:}&#4
\end{tabular}}}
\end{center}}

\title{Two Lower Bounds for BPA}
\author[1]{Qiang Yin}
\author[2]{Mingzhang Huang}
\author[3]{Chaodong He}
\affil[1]{Beihang University, China\\
  BASICS, Shanghai Jiao Tong University, China\\
  \texttt{yinqiang@buaa.edu.cn}}
\affil[2]{BASICS, Shanghai Jiao Tong University, China\\
  \texttt{mingzhanghuang@gmail.com}}
\affil[3]{University of Science and Technology of China \\
  \texttt{hcd@ustc.edu.cn}}
\authorrunning{J.\,Q. Open and J.\,R. Access} %mandatory. First: Use abbreviated first/middle names. Second (only in severe cases): Use first author plus 'et. al.'

\Copyright{John Q. Open and Joan R. Access}%mandatory, please use full first names. LIPIcs license is "CC-BY";  http://creativecommons.org/licenses/by/3.0/

 %Dummy classification -- please refer to \url{http://www.acm.org/about/class/ccs98-html}}% mandatory: Please choose ACM 1998 classifications from http://www.acm.org/about/class/ccs98-html . E.g., cite as "F.1.1 Models of Computation".
% mandatory: Please provide 1-5 keywords
\EventEditors{John Q. Open and Joan R. Acces}
\EventNoEds{2}
\EventLongTitle{42nd Conference on Very Important Topics (CVIT 2016)}
\EventShortTitle{CVIT 2016}
\EventAcronym{CVIT}
\EventYear{2016}
\EventDate{December 24--27, 2016}
\EventLocation{Little Whinging, United Kingdom}
\EventLogo{}
\SeriesVolume{42}
\ArticleNo{23}
\begin{document}

\maketitle

\begin{abstract}
Branching bisimilarity on normed Basic Process Algebra (BPA) was claimed to be EXPTIME-hard in
previous papers without any explicit proof. Recently it is reminded by Jan\v{c}ar that the
claim is not so dependable. In this paper, we develop a new complete proof for EXPTIME-hardness
of branching bisimilarity on normed BPA. We also prove the associate regularity problem on
normed BPA is PSPACE-hard and in EXPTIME. This improves previous P-hard and NEXPTIME result.
\end{abstract}
%\input{introduction}
% !TeX root = ./Two-Lower-Bounds-for-BPA.tex
\section{Introduction}
% survey on results of equivalence checking and regularity checking
Equivalence checking is a core issue of system verification. It
asks whether two processes are related by a specific equivalence.
Baeten, Bergstra and Klop~\cite{BBK87} proved the remarkable result that strong bisimilarity
checking is decidable on normed Basic Process Algebra (BPA).  Extensive work has been aroused since
their seminal paper, dealing with  decidability or complexity issues of checking
bisimulation equivalence on various infinite-state systems (\cite{KuceraJancar2006survey}
for a survey and \cite{Srba2004Roadmap} for an updated overview on this topic).  BPA is a
basic infinite-state system model, and the decidability of weak bisimilarity on BPA is
one of the central open problems.

Fu proved branching bisimilarity, a standard refinement of weak bisimilarity,
is decidable on {normed} BPA~\cite{Fu2013BPA}. Fu also showed the decidability of associate
regularity checking problem in the same paper, which asks
if there exists a finite-state process branching bisimilar to a given normed BPA process.
Recently Jan\v{c}ar and Czerwi{\'{n}}ski improved both decidability results to
NEXPTIME \cite{CzerwinskiJancar2015}, while He and Huang further showed the branching
bisimilarity can actually be decided in EXPTIME \cite{HeHuang2015}.
Both equivalence checking and regularity checking w.r.t. branching
bisimilarity on BPA are EXPTIME-hard \cite{Kiefer2013, Mayr2005}.

Nevertheless, on the \emph{normed} subclass of BPA, the lower bounds for both problems w.r.t. branching bisimilarity
are much less clear. For weak bisimilarity checking on normed BPA, there are several lower bound results in the literature.
St\v{r}{\'{i}}brn{\'{a}} first gave a NP-hard~\cite{StriAbrna1998a} result by reducing from the Knapsack Problem;
Srba then improved it to PSPACE-hard~\cite{Srba2002inf} by reducing from QSAT (Quantified SAT);
and the most recent lower bound is EXPTIME-hard
proved by Mayr~\cite{Mayr2005} by a reduction from the accepting problem of alternation
linear-bounded automaton.

We can verify that all these constructions do not work for branching bisimilarity.
The key reason is that they all make use of multiple state-change internal actions in one process to match certain action in the other.
We previously claimed~\cite{Fu2013BPA, HeHuang2015} that a slight modification on
Mayr's construction~\cite{Mayr2005} is feasible for branching bisimilarity.
However, recently in~\cite{Jancar2016} and in a private correspondence with Jan\v{c}ar, we learned
that all these modifications do not work on \emph{normed} BPA.
The main challenge is  to define state-preserving internal action
sequence structurally in normed BPA. The construction becomes tricky under the normedness condition.
For the same reason, it is
hard to adapt the previous constructions for NP-hard~\cite{StriAbrna1998a} and
PSPACE-hard~\cite{Srba2002inf} results to branching bisimilarity. It turns out that the
lower bound of branching bisimilarity on normed BPA is merely P-hard
\cite{Balcazar1992}. The same happens to regularity checking on normed BPA. The only
lower bound for regularity checking w.r.t. branching bisimilarity is P-hard~\cite{Balcazar1992,
  Srba2003Complexity}. Comparatively, it is PSPACE-hard for the problem w.r.t. weak bisimilarity~\cite{Srba2002inf, Srba2003Complexity}.

In this paper, we study the lower bounds of the equivalence and regularity
checking problems w.r.t. branching bisimilarity on normed BPA. Our contributions are threefold.
\begin{figure}[htbp]
  \centering
  \begin{tabular}{|c||c|c|}
    \hline
    &  Weak Bisimilarity & Branching Bisimilarity \\ \hline \hline
    \multirow{2}{*}{Normed BPA (EC)} & \multirow{2}{*}{EXPTIME-hard~\cite{Mayr2005}}&$\in$ EXPTIME~\cite{HeHuang2015}\\
    &   & \textbf{EXPTIME-hard}\\ \hline
    \multirow{2}{*}{BPA (EC)} &  \multirow{2}{*}{EXPTIME-hard~\cite{Mayr2005}}  &\multirow{2}{*}{EXPTIME-hard~\cite{Kiefer2013}}\\
    &   & \\ \hline
    \multirow{2}{*}{Normed BPA (RC)} &
                                       \multirow{2}{*}{PSPACE-hard~\cite{Srba2002inf,Srba2003Complexity}}  & {\bf $\in$ EXPTIME}  \\
    & & \textbf{PSPACE-hard}\\ \hline
    BPA (RC) & EXPIME-hard~\cite{Mayr2005} &EXPTIME-hard~\cite{Mayr2005}\\
    \hline
  \end{tabular}
  \caption{Equivalence checking and regularity checking on BPA}
  \label{fig:results}
\end{figure}
\begin{itemize}
\item We give the first complete proof on the EXPTIME-hardness of equivalence checking
  w.r.t. branching bisimilarity by reducing from Hit-or-Run game \cite{Kiefer2013}.
\item We show regularity checking w.r.t. branching bisimilarity is PSPACE-hard by a reduction
  from QSAT. We also show this problem is in EXPTIME, which improves previous NEXPTIME result.
\item Our lower bound constructions also work for all equivalence that
  lies between branching bisimilarity and weak bisimilarity, which implies EXPTIME-hardness and
  PSPACE-hardness lower bounds for the respective equivalence checking and regularity checking problems on normed BPA.
\end{itemize}
Fig. \ref{fig:results} summarizes state of the art in equivalence checking (EC) and
regularity checking (RC)  w.r.t. weak and branching bisimilarity on BPA.
The results proved  in this paper are marked with boldface.
In order to prove the two main lower bounds, we study the structure of redundant sets on
normed BPA. Redundant set was first introduced by Fu~\cite{Fu2013BPA} to establish the
decidability of branching bisimilarity. Generally the number of redundant sets of a normed BPA
system is exponentially large. This is also the only exponential factor in the branching
bisimilarity checking algorithms developed in the previous works~\cite{CzerwinskiJancar2015, HeHuang2015}.
Here we would fully use this fact to design our reduction.

%% Orgnization
The rest of the paper is organized as follows. Section \ref{sec:pre} introduces some basic notions; section
\ref{sec:eq_lowerbound} proves the EXPTIME-hardness of equivalence checking; section
\ref{sec:reg_lowerbound} proves the PSPACE-hard lower bound for regularity checking; section
\ref{sec:exptime_upper_bound} proves the EXPTIME upper bound for branching regularity checking; finally
section \ref{sec:conclusion} concludes with some remarks.
%%% Local Variables:
%%% mode: latex
%%% TeX-master: "Two-Lower-Bounds-for-BPA"
%%% TeX-engine: xetex
%%% End:

%\input{preliminary}
% !TeX root = ./Two-Lower-Bounds-for-BPA.tex
\section{Preliminary}\label{sec:pre}
\subsection{Basic Definitions}
A BPA system is a tuple $\Delta=(\mathcal{V},\mathcal{A},\mathcal{R})$, where $\mathcal{V}$ is a
finite set of \emph{variables} ranged over by $A,B,C,\dots,X,Y,Z$; $\mathcal{A}$ is a finite set
of \emph{actions} ranged over by $\lambda$; and $\mathcal{R}$ is a finite set of \emph{rules}. We use a
specific letter $\tau$ to denote {internal action} and use $a,b,c,d,e,f,g$ to range {visible actions} from
the set $\mathcal{A}\backslash\{\tau\}$.
A {process} is a word $w\in \mathcal{V}^{*}$, and will be denoted by $\alpha,\beta,\gamma,\delta,\sigma$. The \emph{nil process} is denoted by  $\epsilon$.
The grammar equivalence is
denoted by $=$. We assume that $\epsilon\alpha=\alpha\epsilon=\alpha$.
A rule in
$\mathcal{R}$ is in the form $X\act{\lambda}\alpha$, where $\alpha$ is a BPA process. The following labeled transition rules
define the operational semantics of the processes.
\[ \inference{X\act{\lambda}\alpha \in \mathcal{R}}{X\act{\lambda}\alpha} \quad \quad \inference{\alpha\act{\lambda}\alpha'}{\alpha\beta\act{\lambda}{\alpha'\beta}}\]
We would write $\alpha\act{}\beta$ instead of $\alpha\act{\tau}\beta$ for simplicity and use
$\Act{}$ for the reflexive transitive closure of $\act{}$.
A BPA process $\alpha$ is
\emph{normed} if there are $\lambda_1$, $\lambda_2$, \dots $\lambda_k$ and $\alpha_1$,
$\alpha_2$,\dots $\alpha_k$ s.t.
$\alpha\act{\lambda_1}\alpha_1\act{\lambda_2}\dots \act{\lambda_k}\alpha_k=\epsilon$. The norm
of $\alpha$, notation $\norm{\alpha}$, is the shortest length of such sequence to $\epsilon$. A BPA system is
normed if every variable is normed. For each process $\alpha$, we use $\textsc{Var}(\alpha)$ to represent the set of
variables that occur in $\alpha$; and we use $|\alpha|$ to denote the size of $\alpha$, which is
define to be the length of word $\alpha$. The size of each rule $X\act{\lambda}\alpha$ is
defined to be $|\alpha|+2$. The size of description of the BPA system $|\Delta|$ is defined by
$|\Delta| = |\mathcal{V}| + |\mathcal{A}| + |\mathcal{R}|$, where $|\mathcal{V}|$ is the number of
variables, $|\mathcal{A}|$ is the number of actions and $|\mathcal{R}|$ is the sum of the size of
all rules. The definition of branching bisimulation is as follows.
\begin{definition}\label{def:beq}
A symmetric relation $R$ on BPA processes is a \emph{branching bisimulation} if whenever
$\alpha R \beta$ and $\alpha\act{\lambda}\alpha'$ then one of the statements is valid:
\begin{itemize}
\item $\lambda=\tau$ and $\alpha' R \beta$;
\item $\beta\Act{}\beta''\act{\lambda}\beta'$ for some $\beta'$ and $\beta''$
  s.t. $\alpha R \beta''$ and $\alpha' R \beta'$.
\end{itemize}
\end{definition}

If we replace the second item in Definition \ref{def:beq} by the following one
\begin{center}
$\beta\Act{}\gamma_1\act{\lambda}\gamma_2\Act{}\beta'$ for some $\gamma_1,\gamma_2$ and $\beta'$ s.t. $\alpha' R \beta'$
\end{center}
then we get the definition of \emph{weak bisimulation}.
The largest branching bisimulation, denoted by $\beq$, is branching bisimilarity; and the
largest weak bisimulation, denoted by $\weq$, is weak bisimilarity.  Both $\beq$ and $\weq$ are
equivalence and are congruence w.r.t. composition. We say $\alpha$ and $\beta$ are branching
bisimilar (weak bisimilar) if $\alpha\beq\beta$ ($\alpha\weq\beta$). A process $\alpha$ is a
\emph{finite-state} process if the reachable set
$\{\beta \mid \alpha\act{\lambda_1}\alpha_1\act{\lambda_2}\dots \act{\lambda_k}\alpha_k=\beta$ and
$ k\in \mathbb{N}\}$ is finite.  Given an equivalence relation $\asymp$, we say a BPA process
$\alpha$ is \emph{regular} w.r.t. $\asymp$, i.e. is \reg{$\asymp$}, if $\alpha\asymp \gamma$ for
some finite-state process $\gamma$. Note that $\alpha$ and $\gamma$ can be defined in different systems.

In this paper we are interested in the
equivalence checking and regularity checking problems  on \emph{normed} BPA. They
are defined as follows, assuming $\asymp$ is an equivalence relation.
\nprob{Equivalence Checking w.r.t. $\asymp$}{A normed BPA system
  $(\mathcal{V},\mathcal{A},\mathcal{R})$ and two processes $\alpha$ and $\beta$}{Is that $\alpha\asymp\beta$ ?}
\nprob{Regularity Checking w.r.t. $\asymp$}{A normed BPA system
  $(\mathcal{V},\mathcal{A},\mathcal{R})$ and a process $\alpha$}{Is that $\alpha$ \reg{$\asymp$}?}

\subsection{Bisimulation Game}
Bisimulation relation has a standard game characterization~\cite{Thomas1993,Stirling1998joy} which is very useful for studying the lower
bounds. A branching (resp. weak) bisimulation game
is a 2-player game played by \emph{Attacker} and \emph{Defender}. A configuration of the game is pair of
processes $(\alpha_0,\alpha_1)$. The game is played in rounds. Each round has 3
steps. First Attacker chooses a move; Defender then responds to match Attacker's move; Attacker then set
the next round configuration according to Defender's response. One round of\emph{ branching bisimulation
game} is defined as follows, assuming $(\beta_0,\beta_1)$ is the  configuration of the
current round:
\begin{enumerate}
\item Attacker chooses $i\in \{0,1\}$, $\lambda$, and $\beta_i'$ to play $\beta_i\act{\lambda}\beta_i'$;
\item Defender responds with $\beta_{1-i}\Act{}\beta_{1-i}''\act{\lambda}\beta_{1-i}'$ for
  some $\beta_{1-i}''$ and $\beta_{1-i}'$. Defender can also play empty response if
  $\lambda=\tau$ and we let $\beta_{1-i}'=\beta_{1-i}$ if Defender play empty response.
\item Attacker set either $(\beta_i',\beta_{1-i}')$ or $(\beta_i, \beta_{1-i}'')$ to the
  next round configuration if Defender dose not play empty response. Otherwise The next round
  configuration is $(\beta_i',\beta_{1-i}')$ automatically.
\end{enumerate}
A round of \emph{weak bisimulation game} differs from the above one at the last 2 steps:
\begin{enumerate}\addtocounter{enumi}{1}
\item Defender responds with $\beta_{1-i}\Act{}\gamma\act{\lambda}\gamma'\Act{}\beta_{1-i}'$ for some
  $\gamma$, $\gamma'$ and $\beta_{1-i}'$; Defender can also play empty response when $\lambda=\tau$ and we let
  $\beta_{1-i}'=\beta_{1-i}$.
\item Attacker set the configuration of next round to be $(\beta_{i}', \beta_{1-i}')$.
\end{enumerate}

If one player gets stuck, the other one wins. If the game is played for infinitely many rounds, then Defender
wins. We say a player has a \emph{winning strategy}, w.s. for short, if he or she can win no matter how
the other one plays. Defender has a w.s. in the branching bisimulation game $(\alpha,\beta)$ iff
$\alpha\beq\beta$; Defender has a w.s. in the weak bisimulation game $(\alpha,\beta)$ iff
$\alpha\beq \beta$.
\subsection{Redundant Set}
The concept of \emph{redundant set} was defined by Fu~\cite{Fu2013BPA}. Given a normed BPA system
$\Delta=(\mathcal{V},\mathcal{A},\mathcal{R})$, a redundant set of $\alpha$, notation
$\Rd{\alpha}$, is the set of variables defined by
\begin{equation}
  \label{eq:9}
  \Rd{\alpha} =\{X \mid X\in \mathcal{V}\land X\alpha \beq \alpha\}
\end{equation}
%It is crucial to decidability and algorithm of branching bisimilarity checking on normed BPA
%\cite{Fu2013BPA, CzerwinskiJancar2015, HeHuang2015}.
We cannot tell beforehand whether there exists $\gamma$ such that $R=\Rd{\gamma}$ for a given
$R$.  We write the redundant set as subscribe of $\gamma_R$ to denote a process such that
$\Rd{\gamma_{R}}=R$, if such process exists. Let us recall a very important property of redundant
sets.
\begin{lemma}[\cite{Fu2013BPA}] \label{lem:fu-lemma}
  If $\Rd{\alpha} = \Rd{\beta}$, then $\gamma_1\alpha\beq\gamma_2\alpha$ iff $\gamma_1\beta\beq\gamma_2\beta$.
\end{lemma}
%We will interested in a semantic norm defined as follows.
\begin{definition}[Relative Branching Norm]
  The branching norm of $\alpha$ w.r.t. $\beta$, notation $\bnorm{\alpha}^{\beta}$, is the minimal number $k\in\mathbb{N}$ s.t.
  \begin{equation}
  \alpha\beta \Act{\beq} \alpha_1'\beta \act{\jmath_1} \alpha_1\beta \Act{\beq}\alpha_2'\beta\act{\jmath_2}\alpha_2\beta \dots \Act{\beq}\alpha_k'\beta\act{\jmath_k}\alpha_k\beta\Act{\beq}\beta
  \label{eq:20}
  \end{equation}
  where $\gamma\Act{\beq}\gamma'$ represents $\gamma\Act{}\gamma'$ and $\gamma\beq \gamma'$,
  $\gamma\act{\jmath_i}\gamma'$ represents $\gamma\act{a} \gamma'$ or $\gamma\act{}\gamma'$ and
  $\gamma\not\beq\gamma'$. The \emph{branching norm} of $\alpha$, notation $\bnorm{\alpha}$, is defined
  to be $\bnorm{\alpha}^{\epsilon}$.
\end{definition}
In fact the relative branching norm is relative to the redundant set of the suffix rather than the process due to the
following.

\begin{restatable}{lemma}{lemrnome}\label{lem:rnorm}
  If $\Rd{\gamma_1}=\Rd{\gamma_2}$, then for any process $\alpha$, $\bnorm{\alpha}^{\gamma_1}= \bnorm{\alpha}^{\gamma_2}$
\end{restatable}
\noindent By Lemma \ref{lem:rnorm}, we will write $\bnorm{\alpha}^{R}$ for $\bnorm{\alpha}^{\gamma_{R}}$.
Some other useful properties of the relative branching norm are as follows: (1) $\alpha\beq\beta$
implies $\bnorm{\alpha} = \bnorm{\beta}$; (2)
$\bnorm{\alpha\beta} = \bnorm{\beta} + \bnorm{\alpha}^{\Rd{\beta}}$; (3) $\bnorm{\alpha}^{\Rd{\beta}} > 0 $
if $\alpha\beta\not\beq\beta$; (4) $\bnorm{\alpha} \leq \norm{\alpha}$.
%%% Local Variables:
%%% mode: latex
%%% TeX-master: "Two-Lower-Bounds-for-BPA"
%%% TeX-engine: xetex
%%% End:

%\input{beq_lowerbound}
% !TeX root = ./Two-Lower-Bounds-for-BPA.tex
\section{EXPTIME-hardness of Equivalence Checking}\label{sec:eq_lowerbound}
% Mary proved that weak
% bisimilarity on nBPA is EXPTIME-hard by reduction from the acceptance problem of alternating
% linear bound automaton \cite{Mayr2005}. It has been claimed in \cite{Fu2013BPA,HeHuang2015}, that by
% slight modification of Mayr's construction, EXPTIME-hardness of branching bisimilarity can also be
% achieved. As no no actual proof is provided, there have been some doubts of this claim
% \cite{Jancar2016}.
In this section, we show that branching bisimilarity on normed BPA is
EXPTIME-hard by a reduction from \emph{Hit-or-Run} game~\cite{Kiefer2013}.
A Hit-or-Run game is a counter game
defined by a tuple $\mathcal{G}=(S_0, S_1, \rightarrow, s_{\vdash},s_{\dashv}, k_{\dashv})$, where
$S=S_0\uplus S_1$ is a finite set of states,
$\rightarrow \subseteq S\times \mathbb{N}\times (S\cup\{s_{\dashv}\})$ is a finite set of
transition rules, $s_{\vdash}\in S$ is the initial state, $s_{\dashv}\not\in S$ is the final
state, and $k_{\dashv}\in\mathbb{N}$ is the final value.  We will use $s\act{\ell}t$ to denote
$(s,\ell,t)\in \rightarrow$ and require that $\ell=0$ or $\ell=2^k$ for some $k$.
For each $s\in S$ there is at least one rule
$(s,\ell, t) \in \rightarrow$.  A configuration of the game $(s,k)$ is an element belong to the set
$(S\cup \{s_{\dashv}\})\times \mathbb{N}$. There are two players in the game, named Player $0$ and Player $1$.
Starting from the configuration $(s_{\vdash},0)$, $\mathcal{G}$ proceeds according to the following rule:
at configuration $(s,k)\in S_i\times\mathbb{N}$, Player $i$ chooses a rule of the form $s\act{\ell}t$ and
increase the counter with $\ell$ and configuration becomes $(t,k+\ell)$. If $\mathcal{G}$ reaches
configuration $(s_{\dashv},k_{\dashv})$ then Player $0$ wins; if $\mathcal{G}$ reaches $(s_{\dashv},k)$
and $k\neq k_{\dashv}$ then Player $1$ wins; if $\mathcal{G}$ never reaches final state, then Player
$0$ also wins.
As a result Player $0$'s goal is to \emph{hit} $(s_{\dashv},k_{\dashv})$ or \emph{run from} the
final state $s_{\dashv}$. Hit-or-Run game was introduced by Kiefer~\cite{Kiefer2013} to establish the
EXPTIME-hardness of strong bisimilarity on general BPA. The problem of deciding the winner of
Hit-or-Run game is EXPTIME-complete with all numbers encoded in binary.
 The main technique result of
the section is as follows.

\begin{restatable}{proposition}{prphitorrunreduction}
% \begin{proposition}
\label{prp:hit_or_run_reduction}
Given a Hit-or-Run game $\mathcal{G}=(S_0, S_1, \rightarrow, s_{\vdash},s_{\dashv}, k_{\dashv})$
we can construct a normed BPA system $\Delta_1 = (\mathcal{V}_1, \mathcal{A}_1, \mathcal{R}_1)$ and two
processes $\gamma$ and $\gamma'$ in polynomial time  s.t.
\begin{center}
  Player $0$ has a w.s. $\iff \gamma \beq \gamma'\iff \gamma \weq \gamma'$.
\end{center}
% \end{proposition}
\end{restatable}
\noindent As a result we have our first lower bound.
\begin{theorem}\label{thm:exphard}
Equivalence checking w.r.t. all equivalence $\asymp$ s.t. $\beq \subseteq \asymp \subset \weq$
on normed BPA is EXPTIME-hard.
\end{theorem}
An EXPTIME algorithm for branching bisimilarity checking on normed BPA is presented in
\cite{HeHuang2015}. By Theorem \ref{thm:exphard}, we have the following.
\begin{corollary}
Branching bisimilarity checking on normed BPA is EXPTIME-complete.
\end{corollary}
Now let us fix a Hit-or-Run game
$\mathcal{G}=(S_0, S_1, \rightarrow, s_{\vdash},s_{\dashv}, k_{\dashv})$ for this section.
The rest part of this section is devoted to proving Proposition \ref{prp:hit_or_run_reduction}
and is organized as follows: in section \ref{sec:sub:binary_number} we introduce a scheme to
represent a $n$-bits binary counter in normed BPA; in section \ref{sec:sub:binary_mnpulate} we give a way
to manipulate this binary counter by branching bisimulation games; and in section
\ref{sec:sub:reduction}, we present the full detail of our reduction and  prove its correctness.

\subsection{Binary Counter Representation}\label{sec:sub:binary_number}
A $n$-bits counter is sufficient for our purpose, where
$n=\lfloor\log_2k_{\dashv}\rfloor +1$. When the counter value is greater than $2^n-1$, Player
$0$ or Player $1$'s object is then to avoid or reach the final state $s_{\dashv}$ respectively
and the exact value of the counter  no longer matters. Another observation is that there are generally
$\exp(|\mathcal{V}|)$ many of redundant sets in a normed BPA system $\Delta=(\mathcal{V},\mathcal{A},\mathcal{R})$.
The main idea of our construction is to use a redundant set of size $n$ to represent the value of a $n$-bits binary counter. The only
challenge part is then to define a proper structure of redundant sets so that it is
fit to be manipulated by branching bisimulation games.
The normed BPA system
$\Delta_0 = (\mathcal{B}\uplus\mathcal{B}', \mathcal{A}_0, \mathcal{R}_0)$ defined as follows fulfills our
objective. % It defines a $n$-bits binary system.
\begin{eqnarray*}
    \mathcal{B} & = & \{ B_i^0, B_i^1 ~|~ 1 \leq i \leq n\}\\
    \mathcal{B}' & = & \{ Z_i^0, Z_i^1, B_i^b(j,b') ~|~ (i\neq j) \land  1\leq i, j \leq n \land  b, b'\in\{0,1\}\} \\
    \mathcal{A}_0 & = & \{d,\tau\} \uplus \{ a_i^0, a_i^1 ~|~ 1 \leq i \leq n\}
\end{eqnarray*}
  and $\mathcal{R}_0$ contains the following rules, where $1\leq i,j,j'\leq n$ and $b,b', b''\in\{0,1\}$:
\begin{center}
\begin{tabular}[center]{|lll|}\hline
      $Z_i^b \act{a_i^b} \epsilon$, & \ \ \ $Z_i^b \act{}\epsilon$; & \\
      $B_i^b \act{a_i^b} B_i^b$,& \ \ \ $B_i^b\act{d}\epsilon$, & \ \ \
                                                                  $B_i^b\act{a_j^{b'}} B_i^b(j,b')$; \ \ $\ (j\neq i)$\\
      $B_i^b(j,b') \act{a_i^b} B_i^b(j,b')$, &\ \ \ $B_i^b(j,b') \act{d} Z_j^{b'}$ & \ \ \
                                                                                     $B_i^b(j,b')\act{a_{j'}^{b''}} B_i^b(j',b'')$;
                                                                                     \ \
                                                                                     $\ (j\neq i \land j' \neq i)$\\
  \hline
\end{tabular}
\end{center}
Intuitively, $B_i^b$ encodes the information that the $i$-th bit of the counter is $b$. One can
verify the equality easily $\Rd{B_i^b}=\{Z_i\}$.
% \begin{example}\label{eg:3bits}
% $\Delta_1 = ({\cal V}_1, {\cal A}_1, {\cal R}_1)$ defines a 3-bits
% binary counter, where
% \begin{eqnarray*}
%   \mathcal{V}_1 & = & \{Z_i^0, Z_i^1, B_i^0, B_i^1 ~|~ 1 \leq i \leq 3\} \cup \{ B_i^b(j,b') ~|~ 1\leq i\neq j \leq 3, b\in\{0,1\}\} \\
%   \mathcal{A}_1 & = & \{d,\tau\} \uplus \{ a_i^0, a_i^1 ~|~ 1 \leq i \leq 3\}
% \end{eqnarray*}
% and $\mathcal{R}_1$ contains the rules showed in Fig. \ref{fig:3bits}.
% \begin{figure}[htbp]
%   \centering
%   \includegraphics[scale=.7]{3bits-color-b.pdf}
%   \caption{Rules for 3 Bits Binary Numbers}
%   \label{fig:3bits}
% \end{figure}
% For space reasons we only show half of all the rules. The other half can be obtained by
% replacing all the blue colored superscript $\textcolor{blue}{0}$ with $1$.
% We can verify the following statements.
% \begin{itemize}
% \item $B_1^1B_2^1B_3^0$ and $B_1^1B_1^0B_2^1B_3^0$ are both encodings of $011$.
% \item $B_1^1B_1^0B_2^1$ and $B_1^1B_2^1B_1^0$ and note valid encodings.
% \item $B_1^0B_2^0B_3^1$ and $B_1^0B_2^0B_3^1B_1^1B_2^1B_3^0$ are both encoding of $100$.
% \end{itemize}
% Note that from $B_1^1B_2^1B_3^0$ to $B_1^0B_2^0B_3^1B_1^1B_2^1B_3^0$ can be seen as increase the
% number $011$ with $1$ and we we get $100$. As all bits are flipped in this case so all updated
% bits, i.e. $B_1^0$, $B_2^0$ and  $B_3^1$, are appended to the head of $B_1^1B_2^1B_3^0$.
%\end{example}

\begin{definition}\label{def:binary_number}
A process $\alpha$ is a \emph{valid encoding} of a $n$-bits binary number $b_nb_{n-1}\dots b_1$,
notation $\alpha\in\enc{b_nb_{n-1}\dots b_1}$, if
$\textsc{Var}(\alpha)\subseteq \mathcal{B}$ and for all $1\leq i \leq n$ there are $\alpha_{i_1}$ and $\alpha_{i_2}$
s.t. $ \alpha = \alpha_{i_1} B_i^{b_i} \alpha_{i_2}$ and $B_i^{1-b_i} \not\in \textsc{Var}(\alpha_{i_1})$.
\end{definition}
If $\alpha\in\enc{b_nb_{n-1}\dots b_1}$, we  use $\sharp\alpha$ to denote the value
$\sum_{i=1}^{n}b_i\cdot 2^{i-1}$. One can see Definition \ref{def:binary_number} as the syntax
of binary numbers in our system $\Delta_0$. This syntax allows us to update a ``binary number'' \emph{locally}.
Suppose $\alpha\in\enc{b_{n}b_{n-1}\dots b_1}$ and we want to flip the $i$-th ``bit''
of $\alpha$ to get another ``binary number'' $\beta$. Then by Definition
\ref{def:binary_number} we can simply let $\beta=B_i^{1-b_i}\alpha$, as one can verify
$\beta \in \enc{b'_nb'_{n-1}\dots b'_1}$ where $b_i'=1-b_i$ and $b_j'=b_j$ for $j\neq i$.
We now give binary numbers a semantic characterization in terms of redundant sets.
%which allows us to test bits of a given binary number with branching bisimulation game.
Let us see a technical lemma first.

\begin{restatable}{lemma}{lembittest}
\label{lem:bit_test}
% \begin{lemma}\label{lem:bit_test}
Suppose $\textsc{Var}(\alpha) \subseteq \mathcal{B}$, the following statements are valid:
\begin{enumerate}
\item \label{lem-11-1}
  $Z_i^b\alpha\beq\alpha$ iff there is $\alpha_1$ and $\alpha_2$
  s.t. $\alpha= \alpha_1B_i^b\alpha_2$ and $B_i^{1-b}\not\in\textsc{Var}(\alpha_1)$;
\item \label{lem-11-2}
  $Z_i^b\alpha\beq\alpha$ implies $Z_i^{1-b}\alpha\not\beq\alpha$.
\end{enumerate}
% \end{lemma}
\end{restatable}
\begin{proof}
We only prove one direction of  (1) here.
Suppose $\alpha=\alpha_1B_i^b\alpha_2$ and $B_i^{1-b}\not\in\textsc{Var}(\alpha_1)$, we show
$Z_i^b\alpha\beq \alpha$. Clearly
$Z_i^bB_i^b\alpha_2\beq B_i^b\alpha_2$, then it is  sufficient to show in the branching bisimulation game
$(Z_i^b\alpha_1B_i^b\alpha_2, \alpha_1B_i^b\alpha_2)$, Defender has a strategy to reach configuration
$(Z_i^bB_i^b\alpha_2, B_i^b\alpha_2)$. Without loss of generality we can assume
$B_i^b\not\in\textsc{Var}(\alpha_1)$. Let $\alpha_1= B_j^{b'}\alpha_1'$, then $j\neq i$. If
Attacker play $\alpha_1B_i^b\alpha_2\act{\lambda}\beta$ for some $\lambda$ or
$Z_i^b\alpha_1B_i^b\alpha_2\act{}\alpha_1B_i^b\alpha_2$, then Defender has a way to response
so that the configuration of next round has a identical process pair.  Attacker's optimal move is
$Z_i^b\alpha_1B_i^b\alpha_2\act{a_i^b}B_j^{b'}\alpha_1'B_i^b\alpha_2$, Defender has to respond
with $\alpha_1B_i^b\alpha_2\act{a_i^b}B_j^{b'}(i,b)\alpha_1'B_i^b\alpha_2$. Now in the new 
configuration $(B_j^{b'}\alpha_1'B_i^b\alpha_2,B_j^{b'}(i,b)\alpha_1'B_i^b\alpha_2)$, Attacker's
optimal move is to play action $d$ as other actions will either make the configuration 
unchanged or lead the game to a configuration of identical process pair. It follows that
$(Z_i^b\alpha_1'B_i^b\alpha_2, \alpha_1'B_i^b\alpha_2)$ is optimal for Attacker and
Defender. Defender then repeat this strategy until the game reaches $(Z_i^bB_i^b\alpha_2, B_i^b\alpha_2)$.
\end{proof}
% \begin{restatable}{lemma}{lembitsdis}
%   \label{lem:bits_dis}
%   Suppose $\textsc{Var}(\alpha)\subseteq\mathcal{B}$, then $Z_i^b\alpha\beq\alpha$ implies $Z_i^{1-b}\alpha\not\beq\alpha$.
% \end{restatable}

\begin{proposition}[Redundant Set Characterization]
\label{prop:binary_rdset}
Let $\alpha$ be a process s.t. $\textsc{Var}(\alpha)\subseteq \mathcal{B}$,
\begin{equation}
  \label{eq:16}
  \alpha \in \enc{b_nb_{n-1}\dots b_1} \iff \Rd{\alpha} = \{Z_1^{b_1}, Z_2^{b_2}, \dots, Z_n^{b_n}\}
\end{equation}
\end{proposition}
\begin{proof}
By (\ref{lem-11-1}) of Lemma \ref{lem:bit_test} and Definition \ref{def:binary_number},
$\alpha\in \enc{b_nb_{n_1}\dots b_1}$ iff
$\{Z_1^{b_1}, Z_2^{b_2}, \dots Z_n^{b_n}\}\subseteq \Rd{\alpha}$. By
(\ref{lem-11-2}) of Lemma \ref{lem:bit_test} we cannot have both $Z_i^{0}$ and $Z_i^{1}$ in
$\Rd{\alpha}$, it follows that $\{Z_1^{b_1}, Z_2^{b_2}, \dots Z_n^{b_n}\}\subseteq \Rd{\alpha}$
iff $\{Z_1^{b_1}, Z_2^{b_2}, \dots Z_n^{b_n}\} =  \Rd{\alpha}$.
\end{proof}
Proposition \ref{prop:binary_rdset} provides us a way to test a specific ``bit'' with branching bisimulation games.
Suppose $\alpha\in\enc{b_nb_{n-1}\dots b_1}$ and we want to check whether $b_i=b$,
then by Proposition \ref{prop:binary_rdset}, we only need to check if Defender has a w.s. in the branching bisimulation game
$(Z_i^b\alpha,\alpha)$. The following lemma shows we can also do bit test with weak
bisimulation games.

\begin{restatable}{lemma}{lembittestweak}
%\begin{lemma}
\label{lem:bit_test_weak}
Suppose $\textsc{Var}(\alpha)\subseteq \mathcal{B}$, then $Z_i^b\alpha \beq \alpha$ iff
$Z_i^b\alpha \weq \alpha$.
%\end{lemma}
\end{restatable}

\noindent The following Lemma shows us how to test multiple bits. It is a simple consequence when
applying Computation Lemma~\cite{Fu2013BPA} to Proposition \ref{prop:binary_rdset} and Lemma
\ref{lem:bit_test_weak}.

\begin{restatable}{lemma}{corbitsrep}
%\begin{corollary}
\label{cor:bits_rep}
Let $\alpha\in\enc{b_nb_{n-1}\dots b_1}$ and $\beta$ be a process
s.t. $\textsc{Var}(\beta)\subseteq \{Z_1^0, Z_1^1, \dots, Z_n^0, Z_n^1\}$, then the
following statements are valid:
\begin{enumerate}
\item $\beta\alpha\beq\alpha$ iff
  $\textsc{Var}(\beta)\subseteq\{Z_1^{b_1}, Z_2^{b_2}, \dots, Z_n^{b_n}\}$;
\item $\beta\alpha\weq\alpha$ iff
  $\textsc{Var}(\beta)\subseteq\{Z_1^{b_1}, Z_2^{b_2}, \dots, Z_n^{b_n}\}$.
\end{enumerate}
%\end{corollary}
\end{restatable}

\subsection{Binary Counter Manipulation}\label{sec:sub:binary_mnpulate}
Suppose we have a counter $\alpha\in \enc{b_nb_{n-1}\dots b_1}$ and want to
increase it by $2^k$, where $0\leq k <n$. This operation has two possible outcomes.
The counter is either updated to some $\beta\in\enc{b'_nb'_{n-1}\dots b_1'}$ with
$\sharp\beta=\sharp\alpha+2^k$, or overflow if $\sharp\alpha+2^k \geq 2^n$.
A key observation on the construction is that we can update $\alpha$ to $\beta$ \emph{locally}.
Although there are $2^n$ many possible values for $\alpha$, we can write $\beta$ as
$\delta\alpha$ for exactly $n{-}k$ possible $\delta$. Indeed, let $\gamma(k,0)$,
$\gamma(k,1)$\dots $\gamma(k,n{-}k)$ and $\delta(k,0)$, $\delta(k,1)$\dots
$\delta(k,n{-}k)$ be the processes defined by
  \[\begin{array}{rclcrcl}
    \gamma(k,0) & = & Z_{k+1}^0 \quad & \quad \quad  \quad  & \delta(k,0)&=& B_{k+1}^1 \\
    \gamma(k,1) & =  & Z_{k+2}^0Z_{k+1}^1 \quad  &\quad \quad \quad  &  \delta(k,1)&=&B_{k+2}^1B_{k+1}^0 \\
      & \vdots & & \quad \quad \quad & & \vdots & \\
    \gamma(k,n{-}k{-}1)  & = & Z_n^0Z_{n-1}^1\dots Z_{k+1}^1\quad   & \quad \quad \quad &\delta(k,n{-}k{-}1)&=& B_{n}^1B_{n-1}^0\dots B_{k+1}^0\\
    \gamma(k,n{-}k) &  = & Z_{n}^1Z_{n-1}^1\dots Z_{k+1}^1\quad  & \quad\quad \quad  & \delta(k,n{-}k) &=& B_{n}^0B_{n-1}^0\dots B_{k+1}^0 \\
\end{array}
\]
we can divide all $\alpha$ in ${n-k+1}$ classes according to $\gamma(k,0)$, $\gamma(k,1)$\dots
$\gamma(k,n{-}k)$. Intuitively, each $\gamma(k,i)$ encodes the bits that will be flipped when
increasing $\alpha$ with $2^k$, while each $\delta(k,i)$ encodes corresponding effect of that
operation.  Let $i^{*}(k)=\Sigma_{i=0}^{n-k-1}(\Pi_{j=0}^{i}b_{k+1+j})$, we have
$0\leq i^{*}(k) \leq n - k$.
Note that $i^{*}(k)$ is the maximal length of successive bits of $1$ starting from $b_{k+1}$ to
$b_n$. By Lemma~\ref{cor:bits_rep}, $\alpha\beq\gamma(k,i)\alpha$ iff $i=i^{*}(k)$. If $i^{*}(k) < n-k$, then
$\sharp\alpha+2^k < 2^n$. By definition of $\delta(k,i^{*}(k))$ we can let $\beta=\delta(k,i^{*}(k))\alpha$ and  have
$\sharp\beta=\sharp\alpha+2^k$. If $i^{*}(k) = n-k$, then $\sharp\alpha+2^k\geq 2^n$ and increasing $\alpha$
with $2^k$ will overflow. %$\delta(k,n{-k})$ is introduced for syntax reason only.

Using the above idea, we can design a branching bisimulation game to simulate the
addition operation. Given  a tuple $\tilde{p}=(k,N, N', O, O')$, where $0\leq k < n$,
and $N$, $N'$, $O$ and  $O'$  are some predefined processes, we define the set of 
variables $\textsc{Add}(\tilde{p})$ by
\begin{equation}
  \textsc{Add}(\tilde{p}) = \{A(\tilde{p}), A'(\tilde{p}), D(\tilde{p}),  D(\tilde{p},i), C(\tilde{p}, i), C'(\tilde{p},i) \mid 0 \leq i \leq n-k\}\label{eq:2}
\end{equation}
The following rules are for  $\textsc{Add}(\tilde{p})$, where $0\leq i, j \leq n-k$,
\[\begin{array}{ll}
    (A1).\quad A(\tilde{p})\act{c} D(\tilde{p}) &\quad  A(\tilde{p})\act{c} D(\tilde{p},i)\\
    (A2). \quad & \quad A'(\tilde{p})\act{c} D(\tilde{p},i)\\
    (A3). \quad D(\tilde{p})\act{c}C(\tilde{p},i) &\quad  \\
    (A4). \quad D(\tilde{p},i)\act{c}C'(\tilde{p},i) &\quad  D(\tilde{p},i)\act{c}C(\tilde{p},j) \quad (j\neq i)\\
    (A5). \quad C(\tilde{p},i)\act{c} \gamma(k,i) &\quad  C'(\tilde{p},i)\act{c}\epsilon\\
    (A6).\quad C(\tilde{p},i) \act{e} N\delta(k,i)&\quad  C'(\tilde{p},i) \act{e} N'\delta(k,i)
                                                  \quad (0\leq i\leq n{-}k{-}1) \\
    (A7). \quad C(\tilde{p},{n{-}k})\act{e}O&\quad   C'(\tilde{p},{n{-}k})\act{e}O'
  \end{array}
\]
The correctness of the simulation is demonstrated by the following Lemma.

\begin{lemma}\label{prp:add}
Suppose $\alpha\in\enc{b_nb_{n-1}\dots b_1}$ and
$i^{*}(k)=\Sigma_{i=0}^{n-k-1}(\Pi_{j=0}^{i}b_{k+1+j})$, then in the branching bisimulation game
starting from $(A(\tilde{p})\alpha, A'(\tilde{p})\alpha)$
\begin{itemize}
\item if $\sharp\alpha+2^k < 2^n$, then Attacker and Defender's optimal play will lead the
  game to $(N\delta(k,i^{*}(k))\alpha, N'\delta(k,i^{*}(k))\alpha)$ with $\sharp(\delta(k,i^{*}(k))\alpha) = \sharp\alpha+2^k$;
\item if $\sharp\alpha+2^k \geq 2^n$, then Attacker and Defender's optimal play will lead
  the game to $(O\alpha, O'\alpha)$.
\end{itemize}
\end{lemma}
\begin{proof}
  % Let
  % $\textsc{Add}(\tilde{p}) = \{A(\tilde{p}), A'(\tilde{p}), D(\tilde{p}),  D(\tilde{p},i), C(\tilde{p}, i), C'(\tilde{p},i) ~|~ 0 \leq i \leq n-k\} $
  % be the set of process variables.
  % Note that we have $0\leq i^{*}(k) \leq n - k$. Consider the operation of
  % increasing a binary number of the value $\sharp\alpha$ with $2^k$, there are 3 cases:
  % \begin{itemize}
  % \item $i^{*}(k)=0$. The operation will only flip the $b_{k+1}$ from $0$ to $1$.
  % \item $0<i^{*}(k)<n-k$. The operation will flip $b_{k+1}$, $b_{k+2}$\dots$b_{k+i^{*}(k)}$ from $1$
  %   to $0$ and flip $b_{k+1+i^{*}(k)}$ from $0$ to $1$.
  % \item $i^{*}(k)=n-k$. The operation will lead to overflow as all bits from $b_{k+1}$ to
  %   $b_n$ is $1$.
  % \end{itemize}
%\begin{equation}
% \label{eq:4}
% \end{equation}
  Rules (A1)(A2)(A3)(A4) form a classical Defender's Forcing gadget \cite{JancarSrba2008}.
  Defender can use it to force the game from configuration $(A(\tilde{p})\alpha, A'(\tilde{p})\alpha)$
  to any configuration of the form $(C(\tilde{p},i)\alpha, C'(\tilde{p},i)\alpha)$, where
  $0\leq i \leq n-k$. But Defender has to play carefully, as at configuration
  $(C(\tilde{p},i)\alpha, C'(\tilde{p},i)\alpha)$ Attacker can use rule (A5) to initiate bits
  test by forcing the game to configuration $(\gamma(k,i)\alpha,\alpha)$.
  By definition of $\gamma(k,i)$ and
  Lemma \ref{cor:bits_rep}, if $i=i^{*}(k)$ then Defender can survive the bits test as
  $\gamma(k,i^{*}(k))\alpha \beq \alpha$; otherwise Defender will lose during the bits test as
  $\gamma(k,i)\alpha \not\weq \alpha$ for $i\neq i^{*}(k)$. Then Defender's optimal move is to force the
  configuration $(C(\tilde{p},i^{*}(k))\alpha, C'(\tilde{p},i^{*}(k))\alpha)$. Attacker's optimal
  choice is to use rules (A6)(A7) to increase the binary number $\alpha$ with $2^{k}$ or flag
  an overflow error by this operation.  If $i^{*}(k) < n-k$, the game reaches
  $(N\delta(k,i^{*}(k))\alpha,N'\delta(k,i^{*}(k))\alpha)$ by rule (A6). As
  $\delta(k,i^{*}(k))$ encodes the effect of bits change caused by increasing $\alpha$ with
  $2^{k}$, one can verify that $\sharp(\delta(k,i^{*}(k))\alpha) = \sharp\alpha+2^{k}$.
  If $i^{*}(k) = n-k$, game goes to $(O\alpha, O'\alpha)$ by rule (A7).
\end{proof}

\begin{remark}
A process $\alpha$ cannot perform an immediate action if there is no $\beta$ s.t. $\alpha\act{}\beta$.
If we require $N$, $N'$, $O$ and $O'$ cannot perform immediate internal actions, then we can
replace the branching bisimulation game of $(A(\tilde{p})\alpha, A'(\tilde{p})\alpha)$ with weak
bisimulation game $(A(\tilde{p})\alpha, A'(\tilde{p})\alpha)$ in Lemma \ref{prp:add}.
\end{remark}

\subsection{The Reduction}\label{sec:sub:reduction}
We assemble the components introduced in the previous sections and present our
reduction now. Let us first recall the Hit-or-Run game
$\mathcal{G}=(S_0, S_1, \rightarrow, s_{\vdash},s_{\dashv}, k_{\dashv})$ and let
$OP(s) = \{(\ell,t) ~|~ (s,\ell,t) \in \rightarrow\}$ and
$OP= \bigcup_{s\in S_0\uplus S_1}OP(s)$. The normed BPA system
$\Delta_1=(\mathcal{V}_1, \mathcal{A}_1,\mathcal{R}_1)$
for Proposition \ref{prp:hit_or_run_reduction} is defined as follows.
\begin{eqnarray*}
  \mathcal{V}_1 & = & \mathcal{B} \uplus \mathcal{B}' \uplus \mathcal{C} \uplus \mathcal{M} \uplus \mathcal{O} \uplus \mathcal{F}\\
  \mathcal{A}_1 & = & \mathcal{A}_0 \uplus \{c,e,f,f',g\}\uplus \{a(\ell,t) ~|~ (\ell,t)\in OP\}\\
  \mathcal{R}_1 & = & \mathcal{R}_0 \uplus \mathcal{R}_1'
 \end{eqnarray*}
%Clearly $\Delta_1$ includes $\Delta_0$ as a subsystem.
We define the variable  sets $\mathcal{C}$, $\mathcal{M}$, $\mathcal{O}$ and $\mathcal{F}$ and add
rules to $\mathcal{R}_1'$ in the following.

\vspace{.6em}
\noindent\textbf{($\mathcal{C}$)}. The variable set $\mathcal{C}$ is used to encode the control states of $\mathcal{G}$ and is defined by
\begin{equation}
\mathcal{C}=\{X(s), X'(s), Y(s), Y'(s) ~|~ s \in S \cup \{s_{\dashv}\}\}.\label{eq:18}
\end{equation}
Our reduction uses the branching (resp. weak) bisimulation game $\mathcal{G}'$ starting from $(\gamma,\gamma')$ to mimic the run of
$\mathcal{G}$ from $(s_{\vdash},0)$.  Defender and Attacker play the role of Player $0$ and
Player $1$ respectively. % Let $\mathcal{G}'$ be the corresponding branching (resp. weak) bisimulation game.
For $0\leq k < 2^n$, let $\mathtt{Bin}(k)$
be the unique $n$-bits binary representation of $k$. The reduction will keep the following
correspondence between $\mathcal{G}$ and $\mathcal{G}'$.  If $\mathcal{G}$ reaches configuration
$(s,k)$ with $k< 2^n$, then $\mathcal{G}'$ will reach configuration
$(X(s)\alpha, X'(s)\alpha)$ for some $\alpha\in\enc{\mathtt{Bin}(k)}$; if $\mathcal{G}$ reaches
$(s,k)$ with $k \geq 2^n$, then $\mathcal{G}'$  $(Y(s)\beta, Y'(s)\beta)$
for some  $\beta\in\enc{b_nb_{n-1}\dots b_1 }$.
Intuitively, $Y(s)$ and $Y'(s)$ are used to indicates the counter of $\mathcal{G}$ overflows.
We do not track the exact value of the counter when it overflows. The two processes $\gamma$ and
$\gamma'$  for Proposition \ref{prp:hit_or_run_reduction} are defined by
\begin{equation}
\gamma = X(s_{\vdash})B_n^0B_{n-1}^0\dots B_1^0 \quad \quad \gamma'= X'(s_{\vdash})B_n^0B_{n-1}^0\dots B_1^0 \label{eq:gamma}
\end{equation}
Clearly $(\gamma,\gamma')$ encodes the initial configuration $(s_{\vdash},0)$.

\vspace{.6em}
\noindent \textbf{($\mathcal{M}$)}. The variable set $\mathcal{M}$ is used to manipulate the $n$-bits binary counter as we discussed in
section \ref{sec:sub:binary_mnpulate} and is defined by
$\mathcal{M}=\bigcup_{\tilde{p}\in\mathbb{P}} Add(\tilde{p})$, where
$\mathbb{P}$ is a set of tuples defined by
\begin{equation}
  \label{eq:14}
  \mathbb{P} = \{(k,X(t),X'(t),Y(t),Y'(t)) ~|~ (2^k,t)\in OP \land 0\leq k < n\}.
\end{equation}
For each $Add(\tilde{p})\subseteq \mathcal{M}$, we add the rules  (A1)(A2)\dots(A7) from section
\ref{sec:sub:binary_mnpulate} to $\mathcal{R}_1'$.

\vspace{.6em}
\noindent \textbf{($\mathcal{O}$)}. The variable set $\mathcal{O}$ is used to initiate the counter update operation and is defined by
\begin{equation}
  \label{eq:15}
  \mathcal{O} = \{A(\ell,t), A'(\ell,t) ~|~ (\ell,t)\in OP\}
\end{equation}
 For each pair $(A(\ell,t), A'(\ell,t)) \in \mathcal{O}$, we add the following rules to
 $\mathcal{R}_1'$  according to $\ell$.

 \begin{itemize}
  \item $A(\ell,t)\act{g}X(t)$ and $A'(\ell,t)\act{g}X'(t)$ if $\ell =0 $;
  \item $A(\ell,t)\act{g}Y(t)$and $A'(\ell,t)\act{g}Y'(t)$ if $\ell\geq 2^n$;
  \item $A(\ell,t) \act{g}A(\tilde{p})$ and $A'(\ell,t) \act{g}A'(\tilde{p})$ if
    $0 < \ell < 2^n$, where $\tilde{p}=(\log\ell,X(t),X'(t),Y(t), Y'(t))$.
  \end{itemize}

\noindent \textbf{($\mathcal{F}$)}.  The set $\mathcal{F}$ defined as follows is used to
implement the Defender's Forcing gadgets.
\begin{equation}
  \label{eq:17}
  \mathcal{F} = \{ E_s, F_s, E_s(\ell,t), F_s(\ell,t) ~|~ s\in S_0 \land (\ell,t) \in OP(s)\}
\end{equation}
We add the following rules to  $\mathcal{R}_1'$ for the variables in $\mathcal{C}\cup\mathcal{F}$.

\begin{itemize}
\item If $s\in S_0$, then in $\mathcal{G}$ Player $0$ chooses a pair $(\ell,t)$ from $OP(s)$.  Rules
  (a1)(a2) and rules (a3)(a4) form two Defender's Forcing gadget. They allow Defender to
  choose the next move in $\mathcal{G}'$. Note that $(\ell,t), (\ell',t')\in OP(s)$

  \vspace{.2em}
  \begin{tabular}{lll}
      (a1). \  $X(s) \act{c} E_s$, & \ $X(s) \act{c} E_s(\ell,t)$, & \ $X'(s) \act{c} E_s(\ell,t)$; \\
      (a2). \  $E_s \actx{a(\ell,t)} A(\ell,t)$, & \ $E_s(\ell,t) \actx{a(\ell,t)} A'(\ell,t)$, &\
                                                                                           $E_s(\ell,t) \actx{a(\ell',t')} A(\ell',t')$; \ \  $((\ell',t') \neq (\ell,t))$\\
      (a3). \  $Y(s) \act{c} F_s$, & \ $Y(s) \act{c} F_s(\ell,t)$, & \ $Y'(s) \act{c} F_s(\ell,t)$;\\
      (a4). \  $F_s \actx{a(\ell,t)} Y(t)$, & \  $F_s(\ell,t) \actx{a(\ell,t)} Y'(t)$, & \
                                                                                  $F_s(\ell,t) \actx{a(\ell',t')} Y(t')$. \ \
                                                                                  $( (\ell',t') \neq (\ell,t))$
  \end{tabular}
\item If $s\in S_1$, then in $\mathcal{G}$ Player $1$ chooses a pair $(\ell,t)$ from $OP(s)$. Correspondingly,
  rule (b1) and (b2) let Attacker choose the next move in $\mathcal{G}'$.

  \vspace{.2em}
  \begin{tabular}{ll}
  (b1). \ $X(s)\actx{a(\ell,t)} A(\ell,t)$, & \ $ X'(s)\actx{a(\ell,t)} A'(\ell,t)$; \quad \ $(\ell,t) \in OP(s)$\\
  (b2). \ $Y(s) \actx{a(\ell,t)} Y(t)$, & \  $Y'(s) \actx{a(\ell,t)} Y'(t)$. \quad \quad $(\ell,t) \in OP(s)$
  \end{tabular}
  \vspace{.1em}
  \item  Let $\mathtt{Bin}(k_{\dashv}) = b_nb_{n-1}\dots b_1$. The following two rules for
    $X(s_{\dashv})$ and $X(s_{\dashv})$ are used
    to initiate bits test.

    \vspace{.2em}
     \begin{tabular}{ll}
     (c). \ $X(s_{\dashv}) \act{f} Z_n^{b_n}Z_{n-1}^{b_{n-1}}\dots Z_1^{b_1}$, & \quad$ X'(s_{\dashv})\act{f} \epsilon$.
     \end{tabular}
     \vspace{.1em}
   \item The following two rules are for $Y(s_{\dashv})$ and $Y(s_{\dashv})$.

     \vspace{.2em}
     \begin{tabular}{ll}
     (d). \ $Y(s_{\dashv}) \act{f} \epsilon,$ \quad $Y'(s_{\dashv}) \act{f'}\epsilon$.
     \end{tabular}
\end{itemize}
\begin{proof}[{{Proof of Proposition \ref{prp:hit_or_run_reduction}}}]
Suppose $\mathcal{G}$ reaches $(s,k)$ for some $s\in S_0\uplus S_1$ and $k < 2^n$, then the
configuration of $\mathcal{G'}$ is $(X(s)\alpha, X'(s)\alpha)$ for some
$\alpha\in\enc{\mathtt{Bin}(k)}$. If $s\in S_0$, then Player $0$ chooses a rule
$s\act{\ell}t$ and $\mathcal{G}$ proceeds to $(t,k+\ell)$. The branching
bisimulation (resp. weak) bisimulation $\mathcal{G'}$ will mimic this behavior while keep the
correspondence between $\mathcal{G}$ and $\mathcal{G'}$ in the following way.
Defender has a strategy to push $\mathcal{G}'$ from  $(X(s)\alpha, X'(s)\alpha)$ to configuration $(X(t)\beta, X'(t)\beta)$ for
some $\beta\in\enc{\mathtt{Bin}(k+\ell)}$ if $k+\ell < 2^n$, or to $(Y(t)\alpha, Y'(t)\alpha)$
if $k+\ell \geq 2^n$.  We only discuss the case $s\in S_0$ here. The argument for $s\in S_1$ is
similar.

First by rules (a1)(a2) Defender forces to the configuration
$(A(\ell,t)\alpha, A'(\ell,t)\alpha)$.
If $\ell = 0$, then $\mathcal{G'}$ reaches $(X(t)\alpha, X'(t)\alpha)$.
If $\ell \geq 2^n$, then $\mathcal{G'}$ reaches $(Y(t)\alpha, Y'(t)\alpha)$.
If $0 < \ell < 2^n$, then $\mathcal{G'}$ first reaches
$(A(\tilde{p})\alpha, A'(\tilde{p})\alpha))$.
Now the binary counter in $\mathcal{G}'$ will be updated according to $\ell$. By Lemma \ref{prp:add}, if $k+\ell < 2^n$, then the optimal play of Attacker and Defender will lead $\mathcal{G'}$ to
$(X(t)\beta, X'(t)\beta)$  with $\sharp\beta=\sharp\alpha+\ell$.
If $k+\ell \geq 2^n$, then the optimal configuration for both Attacker and Defender is
$(Y(t)\alpha, Y'(t)\alpha)$. Once $\mathcal{G}$ reaches
a configuration $(s', k')$ for some $k'\geq 2^n$ and $s'\neq s_{\dashv}$, $\mathcal{G'}$ reaches
$(Y(s')\beta, Y'(s')\beta)$ for some $\beta$. By rules (a3)(a4)(b2), $\mathcal{G'}$ will 
track of the shift of control states of $\mathcal{G}$ while keep $\beta$ intact afterward.

If Player $0$ has a strategy to hit $(s_{\dashv},k_{\dashv})$ or run from
$s_{\dashv}$, then Defender can mimic the strategy to push $\mathcal{G}'$ from
$(\gamma,\gamma')$ to configuration $(X(s_{\dashv})\alpha, X'(s_{\dashv})\alpha)$ for some
$\alpha\in \enc{\mathtt{Bin}(k_{\dashv})}$
or let $\mathcal{G'}$ played infinitely. By rule (c) and Lemma \ref{cor:bits_rep},
$X(s_{\dashv})\alpha \beq X'(s_{\dashv})\alpha$. It follows that $\gamma\beq\gamma'$.
If Player $1$ has a strategy to hit a configuration $(s_{\dashv},k)$ for some 
$k\neq k_{\dashv}$, then Attacker can mimic the strategy to push $\mathcal{G}'$ from
$(\gamma,\gamma')$ to $(X(s_{\dashv})\alpha, X'(s_{\dashv})\alpha)$ for some
$\alpha\in\enc{\mathtt{Bin}(k)}$ if $k<2^n$ or to $(Y(s_{\dashv})\beta,Y'(s_{\dashv})\beta)$ for some $\beta$ if
$k\geq 2^n$. By rule (c) and Lemma \ref{cor:bits_rep},
$X(s_{\dashv})\alpha\not\weq X'(s_{\dashv})\alpha$.
By rule (d), $Y(s_{\dashv})\beta\not\weq Y'(s_{\dashv})\beta$. It follows that $\gamma\not\weq\gamma'$.
\end{proof}

%%% Local Variables:
%%% mode: latex
%%% TeX-master: "Two-Lower-Bounds-for-BPA"
%%% TeX-engine: xetex
%%% End:

%\input{reg_lowerbound}
% !TeX root = ./Two-Lower-Bounds-for-BPA.tex
\section{PSPACE-hardness of Regularity Checking} \label{sec:reg_lowerbound}

Srba proved weak bisimilarity  can be reduced to weak regularity~\cite{Srba2003Complexity} under a
certain condition. We can verify his original construction also works for branching
regularity \reg{$\beq$}. % We have the following theorem.

\begin{theorem}[\cite{Srba2003Complexity}]\label{srba-thm}
  Given a normed BPA system $\Delta$ and two normed process $\alpha$ and $\beta$, one can
  construct in polynomial time a new normed BPA system $\Delta'$ and a process $\gamma$ s.t.
  (1). $\gamma$ is \reg{$\weq$} iff $\alpha\weq\beta$ and both
  $\alpha$ and $\beta$ are \reg{$\weq$}; and (2). $\gamma$ is \reg{$\beq$} iff $\alpha\beq\beta$ and both
  $\alpha$ and $\beta$ are \reg{$\beq$}.
\end{theorem}
In order to get a lower bound of branching regularity on normed BPA we only need to prove a lower
bound of branching bisimilarity. Note that we cannot adapt the previous reduction to get an
EXPTIME-hardness result for regularity as $\gamma$ and $\gamma'$ for Proposition \ref{prp:hit_or_run_reduction} are not \reg{$\beq$}.
Srba proved weak bisimilarity is PSPACE-hard~\cite{Srba2002inf} and the two processes for the
construction are \reg{$\weq$}. This implies weak regularity on normed BPA is PSPACE-hard. However, the construction in
\cite{Srba2002inf} does not work for branching bisimilarity. We can fix this problem by using previous redundant sets
construction.
% we can construct two \reg{$\beq$} processes to get a PSPACE-hard lower bound for
% both weak bisimilarity and branching bisimilarity  as the following proposition
% states.
\begin{proposition}\label{prp:qsat}
Given a QSAT formula $\mathcal{F}$, we can construct a normed BPA system
$\Delta_2 = (\mathcal{C}_2, \mathcal{A}_2, \mathcal{R}_2)$ and two process $X_1$ and $X_1'$ satisfies
the following conditions: (1). If $\mathcal{F}$ is true then  $X_1 \beq X_1'$; (2). If
$\mathcal{F}$ is false then $X_1 \not\weq X_2$; (3). $\alpha$ and $\beta$ are both \reg{$\beq$}.
\end{proposition}
Combining Theorem \ref{srba-thm} and Proposition \ref{prp:qsat} we get our second lower bound.
\begin{restatable}{theorem}{thmregpspace}
% \begin{theorem}
  \label{thm:reg_pspace}
Regularity checking w.r.t. all equivalence $\asymp$ s.t. $\beq \subseteq \asymp \subseteq \weq$
on normed BPA is PSPACE-hard.
% \end{theorem}
\end{restatable}

Now let us first fix a QSAT formula
\begin{equation}
 {\cal F} = \forall x_1 \exists y_1 \forall x_2 \exists y_2 \dots \forall x_m \exists y_m. (C_1 \land C_2 \land \dots \land C_n)\label{eq:6}
\end{equation}
where $C_1\land C_2\land \dots \land C_n$ is a conjunctive normal form with boolean variables
$x_1, x_2 \dots x_m$ and $y_1, y_2, \dots y_m$.
Consider the following game interpretation of the QSAT formula $\mathcal{F}$. There are two players, $\mathbb{X}$ and
$\mathbb{Y}$, are trying to give an assignment in rounds to all the variables
$x_1, y_1, x_2, y_2 \dots, x_m, y_m$. At the $i$-th round, player $\mathbb{X}$ first
assign a boolean value $b_i$ to $x_i$ and then $\mathbb{Y}$ assign a boolean value $b_i'$
to $y_i$. After $m$ rounds we get an assignment
$\mathbb{A}=\bigcup_{i=1}^m\{x_i=b_i, y_i=b_i'\}$. If $\mathbb{A}$ satisfies
$C_1\land C_2\land \dots C_n$, then $\mathbb{Y}$ wins; otherwise $\mathbb{X}$
wins. It is easy to see that $\mathcal{F}$ is true iff $\mathbb{Y}$ has a winning strategy.
This basic idea of constructing $\Delta_2$ is to design a branching (resp. weak) bisimulation game
to mimic the QSAT game on $\mathcal{F}$. This method resembles the ideas in the previous works \cite{Srba2002inf,Srba2002BPP,Srba2002BPA}.
The substantial new ingredient in our construction is $\Delta_0$, introduced in section \ref{sec:sub:binary_number}.
$\Delta_2$ contains $\Delta_0$ as a subsystem and uses it to encode (partial)
assignments in the QSAT game. For $i\in \{1,2,\dots, m\}$ and $b\in\{0,1\}$, let
$\alpha(i,b)$ and $\beta(i,b)$ be the processes defined as follows:
\begin{itemize}
\item $\alpha(i,b)=B_{i_1}^1 B_{i_2}^1\dots B_{i_k}^1$. If $b=1$, then $i_1 < i_2 < \dots < i_k$
  are all the indices of clauses in $\mathcal{F}$ that $x_i$ occurs; if $b=0$, then $i_1 < i_2 < \dots < i_k$
  are all the indices of clauses that $\bar{x}_i$ occurs;
\item $\beta(i,b)=B^1_{i_1} B^1_{i_2}\dots B^1_{i_k}$. If $b=1$, then $i_1 < i_2 < \dots < i_k$
  are all the indices of clauses in $\mathcal{F}$ that $y_i$ occurs; if $b=0$, then $i_1 < i_2 < \dots < i_k$
  are all the indices of clauses that $\bar{y}_i$ occurs;
\end{itemize}
% Intuitively, $\alpha(i,1)$ (resp. $\alpha(i,0)$ ) collects all the clauses that the partial assignment
% $\{x_i=1\}$ (resp. $\{x_i=0\}$) satisfies, while $\beta(i,1)$ (resp. $\beta(i,0)$) collects all the clauses
% that $\{y_i=1\}$ (resp. $\{y_i=0\}$) satisfies.
An assignment $\mathbb{A}=\bigcup_{i=1}^{m}\{x_i=b_i, y_i=b_i'\}$ is represented by the process
$\gamma(\mathbb{A})$ defined by
\begin{eqnarray}
  \label{eq:19}
  \gamma(\mathbb{A}) = \beta(m,b_m')\alpha(m,b_m)\dots \beta(1,b_1')\alpha(1,b_1)
\end{eqnarray}
The following lemma tells us how to test the satisfiability of  $\mathbb{A}$
by bisimulation games.

\begin{restatable}{lemma}{corrdsetqsat}
\label{lem:rdset_qsat}
Suppose $\textsc{Var}(\gamma)\subseteq \{B_1^1, B_2^1, \dots, B_n^1\}$, then the following are
equivalent: (1). $Z_1^1Z_2^1\dots Z_n^1\gamma \beq \gamma$;
(2). $Z_1^1Z_2^1\dots Z_n^1\gamma \weq \gamma$;
(3). $\textsc{Var}(\gamma) = \{B_1^1, B_2^1, \dots, B_n^1\}$.
\end{restatable}

The normed BPA system  $\Delta_2 = ( {\cal C}_2, {\cal A}_2, {\cal R}_2)$ for Proposition \ref{prp:qsat} is defined by
\begin{eqnarray*}
  {\cal C}_2 & = & {\cal B} \uplus {\cal B}' \uplus \{X_i, Y_i, Y_i(1), Y_i(2), Y_i(3) \mid 1 \leq i \leq m\} \uplus\{ X_{m+1}, X_{m+1}'\} \\
  {\cal A}_2 & = & {\cal A}_0 \uplus \{ c_0, c_1, e\} \\
  {\cal R}_2 & = & {\cal R}_0 \uplus {\cal R}_2'
\end{eqnarray*}
and ${\cal R}_2'$ contains the following rules, where $1\leq i \leq m$,
\begin{center}
  \begin{tabular}{|llll|}
    \hline
    1.\quad  &$X_i \act{c_0} Y_i \alpha(i,0)$, & $X_i \act{c_1} Y_i \alpha(i,1)$; & \\
    2.\quad  &$X_i' \act{c_0} Y_i' \alpha(i,0)$, &  $X_i' \act{c_1} Y_i' \alpha(i,1)$; & \\
    3.\quad  &$Y_i \act{e} Y_i(1)$, & $Y_i \act{e} Y_i(2)$, & $Y_i \act{e} Y_i(3)$; \\
    4.\quad  & & $Y_i' \act{e} Y_i(2)$, & $Y_i' \act{e} Y_i(3)$;  \\
    5.\quad  & $Y_i(1) \act{c_0} X_{i+1} \beta(i,0)$, & $Y_i(1) \act{c_1} X_{i+1}\beta(i,1)$;& \\
    6.\quad  & $Y_i(2) \act{c_0} X_{i+1}'\beta(i,0)$, & $Y_i(2) \act{c_1} X_{i+1} \beta(i,1)$; & \\
    7.\quad  & $Y_i(3) \act{c_0} X_{i+1} \beta(i,0)$, & $Y_i(3) \act{c_1} X_{i+1}' \beta(i,1)$; & \\
    8.\quad  & $X_{m+1}\act{e} Z_1^1\dots Z_n^1$, & $X_{m+1}' \act{e}\epsilon$; & \\ \hline
  \end{tabular}
\end{center}
\begin{proof}[Proof of Proposition \ref{prp:qsat}]
Clearly both $X_1$ and $X_1'$ are \reg{$\beq$}. Consider the branching (resp. weak) bisimulation
game starting from $(X_1, X_1')$. A round of QSAT game on $\mathcal{F}$ will be simulated
by 3 rounds branching (resp. weak) bisimulation games. Suppose in the $i$-th round of QSAT game, player
$\mathbb{X}$ assign $b_i$ to $x_i$ and then player $\mathbb{Y}$ assign $b_i'$ to $y_i$. Then in
the branching (resp. weak) bisimulation game, Attacker use rule (1) and (2) put $\alpha(i,b_i)$
to the stack in one round; then in the following two rounds, by Defender's Forcing (rule
(3)(4)(5)(6)(7)), Defender pushes $\beta(i,b_i')$ to the stack. In this way, the branching
(resp. weak) bisimulation game reaches a configuration in the form
$(X_{m+1}\gamma(\mathbb{A}), X'_{m+1}\gamma(\mathbb{A}))$  after $3m$ rounds. Here
$\mathbb{A}=\bigcup_{i=1}^{m}\{x_i=b_i, y_i=b_i'\}$ is an assignment that $\mathbb{X}$ and
$\mathbb{Y}$ generates. It follows from Lemma \ref{lem:rdset_qsat} and rule (8) that if $\mathbb{Y}$ has a
w.s. then Defender can use it to guarantee
$X_{m+1}\gamma(\mathbb{A})\beq X'_{m+1}\gamma(\mathbb{A})$; otherwise if $\mathbb{X}$ has a w.s.
then Attacker can use it to make sure $X_{m+1}\gamma(\mathbb{A})\not\weq X'_{m+1}\gamma(\mathbb{A})$.
\end{proof}

%%% Local Variables:
%%% mode: latex
%%% TeX-master: "Two-Lower-Bounds-for-BPA"
%%% TeX-engine: xetex
%%% End:

%\input{regularity}
% !TeX root = ./Two-Lower-Bounds-for-BPA.tex
\section{EXPTIME Upper Bound for Branching Regularity Checking}\label{sec:exptime_upper_bound}
A normed process is not \reg{$\beq$} if the branching norm of its reachable processes is unbounded.
Let us first introduce a directed weighted graph  $G(\Delta)$ that captures all the
ways to increase the branching norm by performing actions. $G(\Delta)= (V(\Delta), E(\Delta), \mathcal{W})$, where
\begin{eqnarray*}
  V(\Delta) & = & \{ (X, R) \mid X\in \mathcal{V} \land R\subseteq \mathcal{V}^0\land \exists\alpha.\Rd{\alpha}=R\} \\
  E(\Delta) & = & \{ ((X_1, R_1),(X_2, R_2)) \mid X_1\act{\lambda}\sigma X_2\delta \in\mathcal{R} \land \Rd{\delta\gamma_{R_1}} = R_2 \}
\end{eqnarray*}
and $\mathcal{W}: E(\Delta)\rightarrow \{0,1\}$ is a weight function defined as follows. For each
edge $e\in E(\Delta)$, if there is some $\delta$ s.t.
\begin{equation}
  \label{eq:22}
  e=((X_1,R_1), (X_2,R_2)) \land X_1 \act{\lambda}\sigma X_2\delta \in \mathcal{R} \land \Rd{\delta\gamma_{R_1}} = R_2 \land \bnorm{\delta}^{R_1} > 0
\end{equation}
then $\mathcal{W}(e)=1$; otherwise $\mathcal{W}(e) = 0$.
Using the EXPTIME branching bisimilarity checking algorithm from
\cite{HeHuang2015} as a black box, we can compute $G(\Delta)$ in $\exp(|\Delta|)$ time.
\begin{restatable}{proposition}{prpexpG}\label{prp:exp_G}
$G(\Delta)$ can be constructed in $\exp(|\Delta|)$ time.
\end{restatable}
\noindent Suppose $\Rd{\alpha}=R_0$ and there is an edge from $(X_0, R_0)$ to $(X_1, R_1)$
with weight $u$, then there is a sequence of actions $w\in \mathcal{A}^{*}$
s.t. $X_0\act{w}X_1\delta$ and $\bnorm{\delta}^{R_0}\geq u$. So a path
from $(X_0, R_0)$ in the following form
in $G(\Delta)$
\begin{equation}
  (X_0, R_0) \act{u_1} (X_1, R_1) \actx{u_2}\dots \actx{u_k} (X_k, R_k)\label{eq:3}
\end{equation}
indicates there are $w_1$, $w_2$, \dots $w_k$ and $\delta_1$, $\delta_2$, \dots, $\delta_k$ s.t.
$X_0\alpha$ can perform the sequence
\begin{equation}
X_0\alpha \act{w_1} X_1\delta_1\alpha \act{w_2}\dots \act{w_k}X_k\delta_k\delta_{k-1}\dots\delta_1\alpha\label{eq:5}
\end{equation}
with $\Rd{\alpha}=R_0$, $\Rd{\delta_i\delta_{i-1}\delta_1\alpha}=R_i$ and $\bnorm{\delta_i}^{R_{i-1}} \geq u_i$ for $1\leq i\leq k$. Now if
there exists $0 \leq i < k$ s.t.
$X_i= X_k$ and $R_i = R_k$ and $\sum_{j=i+1}^ku_j > 0$, then
we call \eqref{eq:3} a \emph{witness path} of irregularity in $G(\Delta)$ for $X_0\alpha$. Indeed for any $m>0$,
$X_0\alpha$ can reaches $X_i(\delta_k\dots\delta_{i+1})^{m}\delta_i\dots \delta_1\alpha$ by repeating
a subsequence of \eqref{eq:5} $m$ times and we have
\begin{equation}
  \label{eq:7}
  \bnorm{X_i(\delta_k\dots\delta_{i+1})^{m}\delta_i\dots\delta_1\alpha} \geq m\bnorm{\delta_k\dots\delta_{i+1}}^{R_i} = m \sum_{j=i+1}^k \bnorm{\delta_j}^{R_{j-1}} \geq m\sum_{j=i+1}^k u_j \geq m
\end{equation}
This implies $X_0\alpha$ is not \reg{$\beq$}.
The following lemma says that each normed $\alpha$ that is \emph{not} \reg{$\beq$} can be certified by
a witness path in $G(\Delta)$.

\begin{restatable}{lemma}{lemnonregcha}
\label{lem:non_reg_cha}
A process $X_1X_2\dots X_k$ is not \reg{$\beq$} iff there exist $1\leq i \leq k$ and a
witness path in $G(\Delta)$ for $X_iX_{i+1}\ldots X_k$.
\end{restatable}
\noindent The proof idea essentially inherits from \cite{Fu2013BPA}.  We omit the detail here.
With the help of Lemma \ref{lem:non_reg_cha} we can prove the following theorem.

\begin{restatable}{theorem}{thmexpreg}
\label{thm:exp-reg}
Regularity checking w.r.t. $\beq$ on normed BPA is in EXPTIME.
\end{restatable}
%%% Local Variables:
%%% mode: latex
%%% TeX-master: "Two-Lower-Bounds-for-BPA"
%%% TeX-engine: xetex
%%% End:

\section{Conclusion}\label{sec:conclusion}
The initial motivation of this paper is to finish the EXPTIME-completeness~\cite{HeHuang2015} of branching
bisimilarity of normed  BPA. 
The new reduction technique involve the inner structure of BPA w.r.t. branching bisimulation. 
The PSPACE-hard lower bound for
regularity checking is a byproduct once we developed the technique. 
Whether it has a PSPACE algorithm for branching regularity on normed
BPA is a natural further question. We believe the answer is positive.

% \subparagraph*{Acknowledgements.}
% I want to thank \dots
%%% uncomment the following in the submitted version
%%
%% Bibliography
%%

%% Either use bibtex (recommended),
%\bibliography{bpa}

% \printbibliography
%% .. or use the thebibliography environment explicitely
%\input{appendix}
% !TeX root = ./Two-Lower-Bounds-for-BPA.tex
\newpage
\appendix

\section{Proof for Section \ref{sec:pre}}
\lemrnome*
\begin{proof}
Suppose $\bnorm{\alpha}^{\gamma_1}=k_1$ and $\bnorm{\alpha}^{\gamma_2}=k_2$, then there are two
transition sequence
\begin{eqnarray}
  \alpha\gamma_1 \Act{\beq} \alpha_1'\gamma_1 \act{\jmath_1} \alpha_1\gamma_1 \Act{\beq}\alpha_2'\gamma_1\act{\jmath_2}\alpha_2\gamma_1 \dots \Act{\beq}\alpha_{k_1}'\gamma_1\act{\jmath_{k_1}}\alpha_{k_1}\gamma_1\Act{\beq}\gamma_1  \label{eq:rnorm1} \\
  \alpha\gamma_2 \Act{\beq} \beta_1'\gamma_2 \act{\jmath_1} \beta_1\gamma_2 \Act{\beq}\beta_2'\gamma_2\act{\jmath_2}\beta_2\gamma_2 \dots \Act{\beq}\beta_{k_2}'\gamma_2\act{\jmath_{k_2}}\beta_{k_2}\gamma_2\Act{\beq}\gamma_2  \label{eq:rnorm2}
\end{eqnarray}
By Lemma \ref{lem:fu-lemma}, if we substitute $\gamma_1$ with $\gamma_2$ we will get a transition sequence that lead
$\alpha\gamma_2$ to $\gamma_2$ with $k_1$ state-change actions. It follows that
$k_2 \leq k_1$.  Similarly we have $k_1 \leq k_2$. It follows that
$\bnorm{\alpha}^{\gamma_1} = \bnorm{\alpha}^{\gamma_2}$.
\end{proof}
\section{Proofs for Section \ref{sec:eq_lowerbound}}
\subsection{Proof of Lemma \ref{lem:bit_test}}
\lembittest*
\begin{proof}%[Proof of Lemma \ref{lem:bit_test}] % the most basic idea should be presents in
  % this proof
  \begin{enumerate}
  \item (``$\Leftarrow$'') We show by induction on $|\alpha_1|$ that if
    $B_i^{1-b}\not\in \textsc{Var}(\alpha_1)$, then $Z_i^b\alpha_1B_i^b\beq \alpha_1B_i^b$. By
    congruence we have $Z_i^b\alpha_1B_i^b\alpha_2 \beq \alpha_1B_i^b\alpha_2$.
\begin{itemize}
\item $|\alpha_1| = 0$. It is routine to verity that the relation
  $\{(Z_i^bB_i^b, B_i^b), (B_i^b, Z_i^bB_i^b)  \}\cup \beq$ is a branching bisimulation.
\item Suppose $|\alpha_1| = k+1$ and
  $B_i^{1-b}\not\in\textsc{Var}(\alpha_1)$. Let $\alpha_1 = B_j^{b'}\alpha_1'$. If $j=i$, we
  must have $b'=b$. By the base case and congruence we have 
  $Z_i^bB_j^{b'}\alpha_1'\beq B_j^{b'}\alpha_1'$. If $j\neq i$, we show that Defender has a
  w.s. in the branching bisimulation game
  $(Z_i^bB_j^{b'}\alpha_1'B_i^b,B_j^{b'}\alpha_1'B_i^b)$. If Attacker play
  $B_j^{b'}\alpha_1'B_i^b\act{\lambda}\beta$, then Defender responds with
  $Z_i^bB_j^{b'}\alpha_1'B_i^b \act{}B_j^{b'}\alpha_1'B_i^b \act{\lambda}\beta$. The configuration
  of the next round will be a pair of syntax identical processes and Defender wins. If Attacker play
  $Z_i^bB_j^{b'}\alpha_1'B_i^b \act{}B_j^{b'}\alpha_1'B_i^b$, then Defender responds with empty
  transition and the configuration of the next round is
  $(B_j^{b'}\alpha_1'B_i^b, B_j^{b'}\alpha_1'B_i^b)$. Defender also
  wins. Attacker's optimal choice is to play
  $Z_i^bB_j^{b'}\alpha_1'B_i^b \act{a_i^b}B_j^{b'}\alpha_1'B_i^b$, Defender then responds with
  $B_j^{b'}\alpha_1'B_i^b \act{a_i^b}B_j^{b'}(i,b)\alpha_1'B_i^b$ and the game continues from
  $(B_j^{b'}\alpha_1'B_i^b, B_j^{b'}(i,b)\alpha_1'B_i^b)$. Now if Attacker chooses to play an
  action $a_j^{b'}$, then Defender simply follows the suit the game configuration dose not
  change. If Attacker chooses to play an action $a_{j'}^{b''}$ for some $j'\neq j$ and $b''\in \{0,1\}$, then after Defender responds the game
  reaches configuration $(B_j^{b'}(j',b'')\alpha_1'B_i^b, B_j^{b'}(j',b'')\alpha_1'B_i^b)$. Attacker's
  optimal choice is to play the action $d$ and after Defender responds the game reaches
  configuration $(Z_i^b\alpha_1'B_i^b,\alpha_1'B_i^b)$.  Now by I.H. we have
  $Z_i^b\alpha_1'B_i^b\beq\alpha_1'B_i^b $. Defender has a w.s. afterward.
\end{itemize}

(``$\Rightarrow$'') Clearly $Z_i^b \not\beq \epsilon$ and for all $\beta$ we have
$Z_i^bB_i^{1-b}\beta\not\beq B_i^{1-b}\beta$. We show that if there is no $\alpha_1$ and
$\alpha_2$ s.t. $B_i^{1-b}\not\in \textsc{Var}(\alpha_1)$ and $\alpha= \alpha_1B_i^b\alpha_2$,
then Attacker has a w.s. in the branching bisimulation game of $(Z_i^b\alpha, \alpha)$.  Let
$\alpha = B_j^{b'}\alpha'$, then by assumption we have $j\neq i$ or $b\neq b'$. If $j=i$, then
it necessary hold that $b'=1-b$. We are done as
$Z_i^bB_i^{1-b}\alpha'\not\beq B_i^{1-b}\alpha'$. If $j\neq i$, Attacker then play
$Z_i^bB_j^{b'}\alpha\act{a_i^b}B_j^{b'}\alpha'$. Defender has to responds with
$B_j^{b'}\alpha'\act{a_i^b}B_j^{b'}(i,b)\alpha'$ and the game configuration becomes
$(B_j^{b'}\alpha', B_j^{b'}(i,b)\alpha')$. Attacker then play
$B_j^{b'}(i,b)\alpha'\act{d}Z_i^b\alpha'$, Defender has to responds with
$B_j^{b'}\alpha'\act{d}\alpha'$ and game goes to configuration $(Z_i^b\alpha',\alpha')$. If
$B_i^{1-b}\in \textsc{Var}(\alpha)$, Attacker can repeat this strategy until a configuration
of the form $(Z_i^bB_i^{1-b}\beta , B_i^{1-b}\beta)$
is reached; otherwise a configuration of the form $(Z_i^b,\epsilon)$
will be reached. Attacker has a w.s. afterward.

\item By (1) , there is $\alpha_1$ and $\alpha_2$
s.t. $\alpha=\alpha_1B_i^b\alpha_2$ and $B_i^{1-b}\not\in\textsc{Var}(\alpha_1)$. Now suppose
otherwise we also have $Z_i^{1-b}\alpha\beq\alpha$. By (1) again, there are
$\alpha_1'$ and $\alpha_2'$ s.t. $\alpha=\alpha_1'B_i^{1-b}\alpha_2'$ and
$B_i^b\not\in\textsc{Var}(\alpha_1')$. Contradiction.
\end{enumerate}
\end{proof}

\subsection{Proof of Lemma \ref{lem:bit_test_weak}}
\lembittestweak*
\begin{proof}
  It is sufficient to show that $Z_i^b\alpha\not\beq\alpha$ then
  Attacker has a w.s. in the weak bisimulation game of $(Z_i^b\alpha,\alpha)$. By Lemma
  \ref{lem:bit_test} there is no $\alpha_1$ and $\alpha_2$
  s.t. $B_i^{1-b}\not\in \textsc{Var}(\alpha_1)$ and $\alpha= \alpha_1B_i^b\alpha_2$. Let
  $\alpha = B_j^{b'}\alpha'$, then by assumption we have $j\neq i$ or $b\neq b'$. If $j=i$, then
  it necessary hold that $b'=1-b$. We are done as
  $Z_i^bB_i^{1-b}\alpha'\not\weq B_i^{1-b}\alpha'$. If $j\neq i$, Attacker then play
  $Z_i^bB_j^{b'}\alpha\act{a_i^b}B_j^{b'}\alpha'$, Defender has to responds with
  $B_j^{b'}\alpha'\act{a_i^b}B_j^{b'}(i,b)\alpha'$ and the game configuration becomes
  $(B_j^{b'}\alpha', B_j^{b'}(i,b)\alpha')$. Attacker then play
  $B_j^{b'}(i,b)\alpha'\act{d}Z_i^b\alpha'$, Defender has to responds with
  $B_j^{b'}\alpha'\act{d}\alpha'$ and game goes to configuration $(Z_i^b\alpha',\alpha')$. If
  $B_i^{1-b}\in \textsc{Var}(\alpha)$, Attacker can repeat  this strategy until a
  configuration of the form $(Z_i^bB_i^{1-b}\beta , B_i^{1-b}\beta)$ is reached; otherwise a configuration of the form
  $(Z_i^b, \epsilon)$ will be reached. Attacker then has a w.s. afterward.
\end{proof}

\subsection{Proof of Lemma \ref{cor:bits_rep}}
\begin{lemma}[Computation Lemma] % this lemma may need to put into the preliminary
\begin{itemize}
\item []
\item If $\alpha\act{}\alpha_1\act{}\alpha_2 \dots \act{}\alpha_k$ and $\alpha\beq\alpha_k$, then for
  all $1\leq i \leq k$ we have $\alpha\beq\alpha_i$;
\item If $\alpha\act{}\alpha_1\act{}\alpha_2 \dots \act{}\alpha_k$ and $\alpha\weq\alpha_k$, then for
  all $1\leq i \leq k$ we have $\alpha\weq\alpha_i$;
\end{itemize}
\end{lemma}

\begin{lemma}[Lemma 5 of \cite{Fu2013BPA}]\label{lem:forcorrdsetqsat}
  Let $X_1X_2\dots X_k\alpha$ be a normed BPA process, then the following statements are valid
  \begin{itemize}
  \item $X_1X_2\dots X_k\alpha\beq \alpha$ iff $X_i\alpha\beq \alpha$ for all $1\leq i \leq k$.
  \item $X_1X_2\dots X_k\alpha\weq \alpha$ iff $X_i\alpha\weq \alpha$ for all $1\leq i \leq k$.
  \end{itemize}
\end{lemma}

\begin{remark}
Lemma \ref{lem:forcorrdsetqsat} was first noticed by Fu in \cite{Fu2013BPA} as a simple
consequence of Computation Lemma. Although Fu's proof only deals with branching bisimilarity, it
can be adapted to weak bisimilarity easily.
\end{remark}

\corbitsrep*
\begin{proof}
By Lemma \ref{lem:forcorrdsetqsat} and Lemma \ref{lem:bit_test_weak} $\beta\alpha\beq \alpha$
iff $\beta\alpha\weq\alpha$. As a result we only need to prove (1). Suppose
$\beta\alpha\beq\alpha$, by Lemma \ref{lem:forcorrdsetqsat}, for each $X\in \textsc{Var}(\beta)$,
$X\alpha\beq \alpha$. By Proposition
\ref{prop:binary_rdset}, we have $\Rd{\alpha}=\{Z_1^{b_1}, Z_2^{b_2}, \dots, Z_n^{b_n}\}$.
It follows that $\textsc{Var}(\beta)\subseteq \{Z_1^{b_1}, Z_2^{b_2}, \dots, Z_n^{b_n}\}$. Now suppose
$\textsc{Var}(\beta) \subseteq \{Z_1^{b_1}, Z_2^{b_2}, \dots, Z_n^{b_n}\}$, then
$\textsc{Var}(\beta) \subseteq \Rd{\alpha}$. By congruence we have $\beta\alpha\beq \alpha$.
\end{proof}

\section{Proofs for Section \ref{sec:reg_lowerbound}}
\subsection{Proof of Theorem \ref{thm:reg_pspace}}
\thmregpspace*
\begin{proof}
Given a QSAT formular $\mathcal{F}$, we construct a normed BPA system $\Delta$ and two process
$\alpha$ and $\beta$ as Proposition \ref{prp:qsat} dose, then
\begin{equation}
  \label{eq:11}
  \mathcal{F} \textrm{ is true } \iff \alpha\beq \beta \iff \alpha\weq \beta \iff \alpha \asymp \beta
\end{equation}
As $\alpha$ and $\beta$ are both \reg{$\beq$}. We use the construction in the proof of Theorem
\ref{srba-thm} to get another normed BPA system $\Delta'$ and a process $\gamma$, then
\begin{itemize}
\item $\alpha\asymp\beta \implies \alpha\beq\beta\implies \gamma \textrm{ is \reg{$\beq$}}
  \implies \gamma \textrm{ is \reg{$\asymp$}}$;
\item $\alpha\not\asymp\beta \implies \alpha\not\weq\beta\implies \gamma \textrm{ is not  \reg{$\weq$}}
    \implies \gamma \textrm{ is not \reg{$\asymp$}}$.
  \end{itemize}

\noindent It follows that $\mathcal{F}$ is true iff $\gamma$ is \reg{$\asymp$}.
\end{proof}

\subsection{Proof of Lemma \ref{lem:rdset_qsat}}
\corrdsetqsat*
\begin{proof}
  By Lemma \ref{lem:forcorrdsetqsat}, $Z_1^1Z_2^1\dots Z_n^1\gamma\beq \gamma$ iff
  $Z_i^1\gamma\beq \gamma$ for $1\leq i \leq n$ and $Z_1^1Z_2^1\dots Z_n^1\gamma\weq \gamma$ iff
  $Z_i^1\gamma\weq \gamma$ for $1\leq i \leq n$. By Lemma \ref{lem:bit_test_weak}, (1) and (2)
  are equivalent. We only need to prove (1) and (3) are equivalent.
  \begin{itemize}
  \item ``$(1)\Rightarrow(3)$''. By Lemma \ref{lem:bit_test}, $Z_i^1\gamma\beq \gamma$ implies
    $B_i^1\in \textsc{Var}(\gamma)$. By Lemma \ref{lem:forcorrdsetqsat} and assumption we have
    $\textsc{Var}({\gamma})=\{B_1^1,B_2^1,\dots, B_n^1\}$.
  \item ``$(3)\Rightarrow(1)$''. As $\gamma$ contains $B_i^0$, by Lemma \ref{lem:bit_test}
    and assumption, for all $1\leq i \leq n$,
    $Z_i^1\gamma\beq \gamma$. $Z_1^1Z_2^1\dots Z_n^1\gamma\beq \gamma$ follows by congruence.
  \end{itemize}
\end{proof}

\section{Proofs for Section \ref{sec:exptime_upper_bound}}
\subsection{Proof of Proposition \ref{prp:exp_G}}
Let us first recall the main theorem proved by He and Huang.
\begin{theorem}[\cite{HeHuang2015}]\label{thm-hh2015}
  Given a normed BPA system $\Delta=(\mathcal{V}, \mathcal{A}, \mathcal{R})$ and two processes
  $\alpha$ and $\beta$, there is an algorithm that runs in
  $\mathrm{poly}(|\alpha|+|\beta|)\cdot\exp(|\Delta|)$ time to decide if $\alpha\beq \beta$.
\end{theorem}

\prpexpG*
\begin{proof}
  Given a normed BPA system $\Delta=(\mathcal{V},\mathcal{A},\mathcal{R})$, let
  $\mathcal{V}=\{X_1, X_2, \dots, X_n\}$ and $\mathcal{V}^0\subseteq \mathcal{V}$ be the set of
  variables that can reach $\epsilon$ via internal actions alone.
  We first construct a tree $\mathcal{T}$ of size $\exp(|\Delta|)$. The root of the tree is
  $(\epsilon, \Rd{\epsilon})$, and each node on $\mathcal{T}$ is of the form $(\alpha,\Rd{\alpha})$.
  $\mathcal{T}$ is constructed in a BFS way as follows. We first compute
  $(\epsilon,\Rd{\epsilon})$. If there is a node $(\alpha,\Rd{\alpha}) \in \mathcal{T}$ that is unmarked,
  we add the nodes $(X_1\alpha, \Rd{X_1\alpha})$, $(X_2\alpha, \Rd{X_2\alpha})$ \dots
  $(X_n\alpha, \Rd{X_n\alpha})$ to $\mathcal{T}$ as the children of $(\alpha, \Rd{\alpha})$ and
  then mark $(\alpha,\Rd{\alpha})$ as ``processed''.
  Now for each new added node $(X_i\alpha, \Rd{X_i\alpha})$, if there is some node $(\beta,\Rd{\beta})\in\mathcal{T}$ that has been
  marked as ``processed'' and $\Rd{\beta} = \Rd{X_i\alpha}$, we then mark
  $(X_i\alpha, \Rd{X_i\alpha})$ as a ``leaf''. The construction $\mathcal{T}$ stops if all nodes are
  marked as either ``processed'' or ``leaf''.
  Clearly there are at most $\exp(|\Delta|)$ number of nodes. And the size of each node is bounded by $\exp(|\Delta|)$.
  By Theorem \ref{thm-hh2015}, we can compute $\Rd{\alpha}$ in
  $|\mathcal{V}^0|\cdot\mathrm{poly}(|\alpha|)\cdot\exp(|\Delta|)$ time for a given process $\alpha$.
  It follows that we can construct $\mathcal{T}$ in $\exp(|\Delta|)$ time.

  We then compute $V(\Delta)$, $E(\Delta)$ and
  $\mathcal{W}: E(\Delta) \rightarrow \{0,1\}$ from $\mathcal{T}$ as follows.
  \begin{enumerate}
  \item  By Lemma \ref{lem:fu-lemma}, a set $R\subseteq \mathcal{V}^{0}$ has some $\alpha$
    with $\Rd{\alpha}=R$ iff there is a node $(\beta, R)\in\mathcal{T}$ for some $\beta$. As a result, we can let
    \begin{eqnarray*}
      V(\Delta) = \{(X, R) ~|~ X \in \mathcal{V} \land  \exists \beta. (\beta, R) \in \mathcal{T} \}
    \end{eqnarray*}
    Clearly $V(\Delta)$ can be computed in $\exp(|\Delta|)$ time.
  \item For each $(X,R)\in V(\Delta)$, we enumerate all the rules of the form
    $X\act{\lambda}\gamma Y\delta$. If there is a $\delta$ with
    $\textsc{Var}(\delta) \not\subseteq R$, then we add an edge $e=((X,R), (Y,R'))$ to
    $E(\Delta)$ with $\mathcal{W}(e)=1$, where $R'=\Rd{\delta\beta}$ and
    $(\beta,R)\in\mathcal{T}$ for some $\beta$. By Lemma \ref{lem:fu-lemma}, $R'$ can be read from
    $\mathcal{T}$. If for all rules of form $X\act{\lambda}\gamma Y\delta$ we have
    $\textsc{Var}(\delta) \subseteq R$, then we add an edge $e=((X,R), (Y,R))$ to $E(\Delta)$
    with $\mathcal{W}(e) = 0$. For each node $(X,R)$, we can compute all the edge from $(X,R)$
    with its weight value in $|\Delta|^2\exp(|\Delta|)$ time. As a result, $E(\Delta)$ and
    $\mathcal{W}$ can be computed in $\exp(|\Delta|)$ time.
  \end{enumerate}
\end{proof}

\subsection{Proof of Lemma \ref{lem:non_reg_cha}}
\lemnonregcha*
\begin{proof}
  Let $\Delta=(\mathcal{V}, \mathcal{A}, \mathcal{R})$ and $\alpha=X_1X_{2}\dots X_k$.
  It is easy to verity that if there is a witness path for some $X_iX_{i+1}\dots X_k$, then for
  any $m>0$ there is some $w_m$
  s.t. $\alpha\act{w_m}\beta_m$ and $\bnorm{\beta_m} \geq m$. It follows that $\alpha$ is not
  \reg{$\beq$}. Now suppose $\alpha$ is not \reg{$\beq$}, then there is some $\beta$
  reachable from $\alpha$ and $\bnorm{\beta} - \bnorm{\alpha} > (|\mathcal{V}|\cdot2^{|\mathcal{V}|}+1)r_{\Delta}\norm{\Delta}+2\norm{\Delta}$, where
  \begin{eqnarray}
    \label{eq:1}
    r_{\Delta} & = & \textrm{max}\{ |\alpha| ~|~ X\act{\lambda}\alpha \in \mathcal{R}\} \\
        \norm{\Delta} & = & \textrm{max}\{\norm{X} ~|~ X\in\mathcal{V}\}
  \end{eqnarray}
  Now Let
  \begin{equation}
    \label{eq:8}
    \alpha = \alpha_0 \act{\lambda_1} \alpha_1 \act{\lambda_2} \dots \act{\lambda_m} \alpha_m = \beta
  \end{equation}
  be a transition sequence that $\alpha$ reaches $\beta$. From \eqref{eq:8} we can compute a
  sequence of indices $0 \leq s_0 < s_1 < \dots < s_k = m $ as follows
  \begin{eqnarray}
    \label{eq:10}
    h_0 & = & \textrm{min}\{ |\alpha_j| ~|~ 0 \leq j \leq m\} \\
    s_0 & = & \textrm{min} \{ j ~|~ |\alpha_j| = h_0\} \\
    h_{i+1} & = & \textrm{min}\{ |\alpha_j| ~|~  s_i < j \leq m\} \\
    s_{i+1} & = & \textrm{min} \{ j ~|~ |\alpha_j| = h_{i+1} \land s_i < j \leq m \}
  \end{eqnarray}
  Let $\alpha_{s_i} = Y_i\beta_i$, by the definition of $s_i$ and $s_{i+1}$, there are
  $\lambda_i\in\mathcal{A}$, $w_i\in\mathcal{A}^{*}$ and
  $\sigma_i,\delta_i \in \mathcal{V}^{*}$
  s.t. $Y_{i} \act{\lambda_{i}} \sigma_i Y_{{i+1}} \delta_i \in \mathcal{R}$, $\sigma_i\act{w_i}\epsilon$ and
  $\delta_i\beta_{{i}} = \beta_{{i+1}}$ for $0\leq i < k$. As a result we can rewrite the subsequence of \eqref{eq:8} from
  $\alpha_{s_0}$ to $\alpha_{s_k}$ by
  \begin{equation}
    \label{eq:24}
    Y_0\beta_0\actx{\lambda_0w_0} Y_1\delta_0\beta_0\actx{\lambda_1w_1}\dots \actx{\lambda_{k-1}w_{k-1}}Y_k\delta_{k-1}\dots\delta_0\beta_0 = Y_k\beta_k
  \end{equation}
 By
  definition of $s_0$ we have $\bnorm{Y_0\beta_0}-\bnorm{\alpha} \leq \norm{\Delta}$ and by
  assumption we have
  $\bnorm{Y_k\beta_k}-\bnorm{\alpha} >  (|\mathcal{V}|\cdot2^{|\mathcal{V}|}+1)r_{\Delta}\norm{\Delta} + 2\norm{\Delta}$,
  thus
  \begin{eqnarray}
    \bnorm{\delta_{k-1}\dots \delta_0}^{\Rd{\beta_0}} & = &(\bnorm{Y_k\beta_k} - \bnorm{Y_0\beta_0}) + (\bnorm{Y_0}^{\Rd{\beta_0}}-\bnorm{Y_k}^{\Rd{\beta_k}}) \\
    & >&  (|\mathcal{V}|\cdot2^{|\mathcal{V}|}+1)r_{\Delta}\norm{\Delta}
  \end{eqnarray}
  On the other hand let $v_i=\bnorm{\delta_i}^{\Rd{\beta_i}}$ for $0\leq i < k$ and we have
  \begin{equation}
  \bnorm{\delta_{k-1}\dots \delta_0}^{\Rd{\beta_0}} = \sum_{i=0}^{k-1}v_i \leq kr_{\Delta}\norm{\Delta}\label{eq:25}
\end{equation}
  It follows that there are $0\leq i_1 < i_2 \leq k$ s.t. $Y_{i_1}=Y_{i_2}$,
  $\Rd{\beta_{i_1}} = \Rd{\beta_{i_2}}$ and $\sum_{j=i_1}^{i_2-1}v_{j} > 0$. Note that the
  transition sequence of \eqref{eq:24} induces the following path in $G(\Delta)$
  \begin{equation}
    \label{eq:13}
    (Y_0, \Rd{\beta_0}) \act{u_0} (Y_1, \Rd{\beta_1}) \act{u_1} \dots \act{u_{k-1}} (Y_k, \Rd{\beta_k})
  \end{equation}
  By definition of $G(\Delta)$ we have $u_i = 1$ if $v_i > 0$. It follows that
  $\sum_{j=i_1}^{i_2-1}u_j > 0$.
  By definition of $s_0$ and $h_0$, we have $Y_0\beta_0=X_{k-h_0+1}\dots X_k$.
  As a result the subpath of \eqref{eq:13} that from $(Y_0, \Rd{\beta_0})$ to
  $(Y_{i_2}, \Rd{\beta_{i_2}})$ is a witness path in $G(\Delta)$ for $X_{k-h_0+1}\dots X_k$.
\end{proof}
\subsection{Proof of Theorem \ref{thm:exp-reg}}
\thmexpreg*
\begin{proof}
Given a normed BPA system $\Delta=(\mathcal{V},\mathcal{A},\mathcal{R})$ and a process $\alpha$
we use the following procedure to decide if $\alpha$ is \reg{$\beq$}.
\begin{enumerate}
\item Construct $G(\Delta)$.
\item Compute the set of growing nodes $V'(\Delta)\subseteq V(\Delta)$. A node $(X,R)$ in
  $V(\Delta)$ is a \emph{growing node} if there is simple circle of the following form in $G(\Delta)$
  \[(X,R) \act{u_1} (Y_1,R_1) \act{u_2}\dots \act{u_{k-1}} (Y_{k-1},R_{k-1}) \act{u_k}(X,R)\]
  and $\sum_{i=1}^{k}u_i > 0$. For each $(X,R)\in V(\Delta)$ we can decide
  whether $(X,R)\in V'(\Delta)$ as follows. Let $B_X^R(\Delta)$ be the set of nodes
  reachable from  $(X,R)$ via a path of total weight $0$. It is necessary that $|B_X^R(\Delta)| \leq |\mathcal{V}|$. Now let
  \[ A_X^R(\Delta) = \{ (Y,R') ~|~ (X',R)\in B_X^R(\Delta) \land ((X',R), (Y,R')) \in E(\Delta) \land \mathcal{W}((X',R), (Y,R')) = 1\}\]
   Clearly $|A_X^R| \leq |\mathcal{V}||\Delta|$ and the computation of $A_X^R$ can done in
   $|G(\Delta)|$ time. We can verify that $(X,R)$ is a growing node iff $(X,R)$ is reachable
   from some $(Y,R')\in A_X^R(\Delta)$.
\item Let $\alpha=X_k\dots X_1$ and $\gamma_i= X_i\dots X_1$ for $1 < i \leq k$; and let
  $\gamma_1=\epsilon$. If there is some $i$ s.t. $(X_i, \Rd{\gamma_i})$ can reach a node in $V'(\Delta)$
  in $G(\Delta)$, then output ``not regular''; otherwise output ``regular''.
\end{enumerate}
By Proposition \ref{prp:exp_G}, step (1) can be done in $\exp(|\Delta|)$ time.  By Theorem
\ref{thm-hh2015} computing $\Rd{\gamma_i}$ can be done in
$\mathrm{poly}(|\alpha|)\cdot\exp(|\Delta|)$ time.  The other part of step (2) and step (3) only
checks reachability properties in $G(\Delta)$, which can be done $|G(\Delta)|^2$ time. Note
$G(\Delta)$ is of $\exp(|\Delta|)$ size. As a result the whole procedure can be done in
$\mathrm{poly}(|\alpha|)\cdot\exp(|\Delta|)$ time.
\end{proof}

%%% Local Variables:
%%% mode: latex
%%% TeX-master: "Two-Lower-Bounds-for-BPA"
%%% TeX-engine: xetex
%%% End:

\end{document}